\documentclass[journal]{IEEEtran}
\usepackage{amsmath}                  
\usepackage{amssymb,amsfonts}         
\usepackage{mathrsfs}                 
\usepackage{graphicx}                 
\usepackage{subfigure}                
\usepackage{caption}                  
\usepackage{float}                    
\usepackage{bm}                       
\usepackage{stfloats}                 
\usepackage{algorithm}
\usepackage{algorithmic}
\usepackage{cite}
\usepackage{color}
\usepackage{cases}
\usepackage{amsthm}
\usepackage{epsfig}
\usepackage{dsfont}
\usepackage{array}
\usepackage{enumerate}

\theoremstyle{definition}
\newtheorem{myDef}{Definition}
\newtheorem{myThe}{Theorem}
\newtheorem{myLem}{Lemma}

\begin{document}
%
\title{Robust Fuzzy-Learning For Partially Overlapping Channels Allocation In UAV Communication Networks}
%
%
%

\author{Chaoqiong~Fan,
        Bin~Li,
        Jia~Hou,
        Yi~Wu,
        Weisi~Guo,
        Chenglin~Zhao
\thanks{Chaoqiong~Fan, Bin Li and Chenglin Zhao are with the School of Information and Communication Engineering (SICE),
Beijing University of Posts and Telecommunications (BUPT), Beijing, 100876, China. (Email: stonebupt@gmail.com).}
\thanks{Jia Hou is with School of Electronics and Information Engineering, Soochow University, Suzhou 215006, China. (Email: houjiastock@hotmail.com.)}
\thanks{Yi~Wu is with Fujian Provincial Engineering Technology Research Center of Photoelectric Sensing Applications, Fujian Normal University, Fuzhou, Fujian 350007, China. (Email: wuyi@fjnu.edu.cn.) }
\thanks{Weisi Guo is with School of Engineering, University of Warwick,
West Midlands, CV47AL. (Email: Weisi.Guo@warwick.ac.uk) }
}

%
%



\maketitle

\begin{abstract}
In this paper, we consider a mesh-structured unmanned aerial vehicle (UAV) networks exploiting partially overlapping channels (POCs). For general data-collection tasks in UAV networks, we aim to optimize the network throughput with constraints on transmission power and quality of service (QoS).
As far as the highly mobile and constantly changing UAV networks are concerned, unfortunately, most existing methods rely on definite information which is vulnerable to the dynamic environment, rendering system performance to be less effective.
In order to combat dynamic topology and varying interference of UAV networks, a robust and distributed learning scheme is proposed.
Rather than the perfect channel state information (CSI), we introduce uncertainties to characterize the dynamic channel gains among UAV nodes, which are then interpreted with fuzzy numbers.
Instead of the traditional observation space where the channel capacity is a crisp reward, we implement the learning and decision process in a mapped fuzzy space. This allows the system to achieve a smoother and more robust performance by optimizing in an alternate space.
To this end, we design a fuzzy payoffs function (FPF) to describe the fluctuated utility, and the problem of POCs assignment is formulated as a fuzzy payoffs game (FPG).
Assisted by an attractive property of fuzzy bi-matrix games, the existence of fuzzy Nash equilibrium (FNE) for our formulated FPG is proved.
Our robust fuzzy-learning algorithm could reach the equilibrium solution via a least-deviation method.
Finally, numerical simulations are provided to demonstrate the advantages of our new scheme over the existing scheme.
\end{abstract}


%
\IEEEpeerreviewmaketitle

\section{Introduction}
%
%
%
%

\IEEEPARstart{T}{riggered} by the development of automation and sensor technology, unmanned aerial vehicles (UAVs) have become increasingly prevalent in military, public and civil applications, such as autonomous combat, target detection, video surveillance, data collection, disaster management, network coverage extension and so on \cite{zeng2016wireless,koulali2016a}.
With the distinctive advantages of high mobility, quick deployment, cost-effective and line-of-sight (LOS) or near LOS communication channels, UAVs open a promising prospect to the future development of society and technology, and therefore, have attracted tremendous interest from both academia and industry \cite{gupta2016survey,jiang2012optimization}.
To support diverse applications of UAVs, reliable and effective information transmission within UAVs network and with the ground control station (GCS) is of crucial importance \cite{goddemeier2012role}. However, communication between a swarm of UAVs can become unreliable especially when considering the high mobility of UAVs, the constrained power of hardware, the throughput performance requirements of data packets, and the limited amount of radio resources \cite{liu2014hawk}.

Radio resources, particularly the available channels in specific wireless spectrum, are generally regarded as a main factor affecting the communication quality \cite{naeem2014resource}.
Besides, the rapid development of wireless technologies poses growing demands for the limited spectrum resource \cite{andrews2014will}.
Therefore, improving the efficiency of channel utilization becomes more significant for emerging wireless services \cite{chao2017learning}, including the UAV networks.
Even for the military UAV applications with adequate spectrum resource, enhancing spectrum efficiency is no doubt of great importance to improve the reliability and effectiveness of data transmission \cite{orfanus2016self}.
To this end, the optimal resources allocation with higher utilization efficiency along with better network performance in UAV systems is naturally a fundamental and important problem to be addressed.
Since the number of non-overlapping channels (orthogonal channels) is limited by the available spectrum, the partially overlapping channels (POCs) have become the focus of research in the past five years \cite{duarte2012partially}.
As one of the most prospective techniques in multi-radio multi-channel (MRMC) field, POCs specified by IEEE 802.11b/g standard can improve the network throughput, by permitting more parallel transmissions under the tolerable interference \cite{cui2011partially}.
As such, the allocation of POCs should be properly designed and, otherwise, adjacent channels and self interference may become serious, by noticeably degrading network performance instead of improving it \cite{yong2016partially}.

There are few studies on optimizing channel/spectrum resource in UAV networks \cite{si2015dynamic, li2017optimal}.
Under a configuration of orthogonal channels, a navigation data-assisted opportunistic spectrum access (OSA) scheme in heterogeneous UAV networks is proposed \cite{si2015dynamic}.
In \cite{li2017optimal}, a resource allocation scheme is presented to minimize mean packet transmission delay in multi-layer UAV networks.
In order to accommodate more parallel transmission channels, the POCs assignment in a combined UAV-D2D network based on the crisp game theory is firstly studied in \cite{tang2017ac}.
Various schemes have been proposed to implement the optimal assignment of POCs in the context of static wireless networks, e.g. the WLAN scenarios \cite{tang2017ac,mishra2006partially,ding2012using,duarte2010partially,li2017throughput,zhao2016dapa,liu2010load,xu2013opportunistic}.
Research in \cite{mishra2006partially} demonstrates that POCs can efficiently avoid interference and improve the overall throughput by proper assignment.
In \cite{ding2012using}, a greedy algorithm is presented for POCs allocation to maximize the network throughput, while a heuristic POCs assignment algorithm is proposed in \cite{duarte2010partially}. \cite{li2017throughput} presentes an interference-tolerant medium access method to optimize the WLAN/cellular integrated network (WCIN) throughput by utilizing POCs.
In \cite{zhao2016dapa} the authors study the problem of interaction between density of access points (APs) and POCs assignment with parameter tuning.
\cite{liu2010load} proposes a load-aware channel assignment exploiting POCs for wireless mesh networks.
The problem of distributed channel allocation in OSA networks with POCs using a game theoretic learning algorithm is investigated in \cite{xu2013opportunistic}.

It is noteworthy that, however, the distinguishing features of UAV networks would compromise the effectiveness of existing methods developed for POCs allocation in ground wireless networks.
First, most schemes rely on an ideal assumption that the channel state information (CSI) is static and can be perfectly estimated.
Unfortunately, due to the high mobility of UAV nodes, the intermittence of links, the dynamics of topologies and the short-training duration, such assumptions will become impracticable in UAV scenarios. Therefore, most schemes may be vulnerable to dynamic environments \cite{fan2017robust}.
Moreover, taking the hardware limitations in UAVs into consideration (e.g. volume and weight), a UAV node is generally energy-constrained \cite{nguyen2013energy}.
Thus, the information exchange within a whole network is resource-demanding and tends to be formidable.
As far as UAV scenarios are concerned, existing learning schemes premised on global and definite knowledge will be no longer reliable, by greatly deteriorating the network performance.
To the best of our knowledge, robust channel allocation for dynamic UAV networks has not been reported in previous works, especially when considering the ubiquitous varying CSI and the complex coupling interferences.

In this paper, we focus on the robust and distributed POCs assignment in UAV communication networks, by fully taking its intrinsic dynamics and uncertainties in to consideration.
To be specific, we aim specially at realizing robust channel accessing in the presence of uncertain environmental knowledge, and simultaneously, guarantee the QoS-provisioning data transmission under time-varying mutual interference \cite{Perlaza2012quality}.
As opposed to the previous crisp game-theoretical approaches, in this work we propose a novel robust fuzzy-learning scheme for distributed channel allocation in UAV communication networks.
By mapping the uncertain utility from a direct observation space to another fuzzy space, and with aid of fuzzy-logic analysis, our proposed method effectively relaxes the sensitiveness to changing environments. It thereby ensures a robust channel accessing and QoS-aware transmissions even in dynamic UAV scenarios. To sum up, the main contributions of this work are summarized as follows.

\begin{enumerate}[(1)]
 \item We formulate an optimization problem of POCs allocation in UAV communication networks.
       Specifically, we consider the mesh-structured UAV network and introduce a virtual \emph{interference factor} to thoroughly characterize the coupling network interferences.
       By taking the QoS-provisioning requirement into account, we model the optimal POCs allocation in highly dynamic environments as one global throughput maximization problem with multiple constraints.
 \item We develop a fuzzy payoffs game (FPG) to describe the optimal POCs allocation with uncertain dynamics, and investigate the property of FPG to ensure the existence of fuzzy Nash equilibrium (FNE).
       In order to alleviate the sensitiveness to varying CSI and coupling interferences, we employ the fuzzy number to describe the uncertain channel gains, and use the fuzzy payoffs function (FPF) to evaluate the fluctuated utility of UAV nodes \cite{li2013effective}.
       Owing to the fuzzy-number representation and the fuzzy-logic computation, our new FPG can effectively address the channel resource competition with dynamic and uncertain environmental information.
 \item We design a robust fuzzy-learning algorithm for distributed POCs allocation.
       For the FPG evolving the fuzzy numbers, we first cope with the fuzzy payoffs in the mapped fuzzy space, and calculate the priority vector of channels with the assistance of fuzzy preference relation (FPR).
       Relying on the priority vector derived in a fuzzy space, the UAV nodes can implement a robust fuzzy learning and, therefore, distributed updating to achieve the FNE of the formulated FPG.
       In this regards, our scheme is capable of combating the dynamic environments and thereby realizing the optimal POCs allocation to maximize the global throughput with the predefined constraints.
 \item We evaluate the performances of our robust fuzzy-learning scheme in the mesh-structured UAV networks.
       Numerical results show that our proposed scheme can achieve the maximum throughput and improve the resources efficiency, even with the dynamic and uncertain utility, whereas the crisp-game based scheme is less effective in terms of robust allocation.
       We demonstrate our new fuzzy learning scheme can significantly improve the global throughput ($>$60\%) and the allowable active-link numbers, which hence provides great promises to emerging UAV communication networks.
\end{enumerate}

The remainder of this paper is organized as follows. In Section II, we formulate the optimal POCs allocation problem with constrains for the mesh-structured UAV networks. In Section III, based on the preliminaries of game theory and fuzzy set theory, we develop a FPG, and further prove the existence of FNE. The robust fuzzy-learning algorithm for distributed POCs allocation in UAV networks is proposed in Section IV. The performances of our proposed scheme are demonstrated via numerical simulations in Section V. Finally, we draw the conclusions of our work in Section VI.

\section{System Model and Problem Formulation}
\begin{figure}[!t]           
 \centering
 \includegraphics[width=80mm]{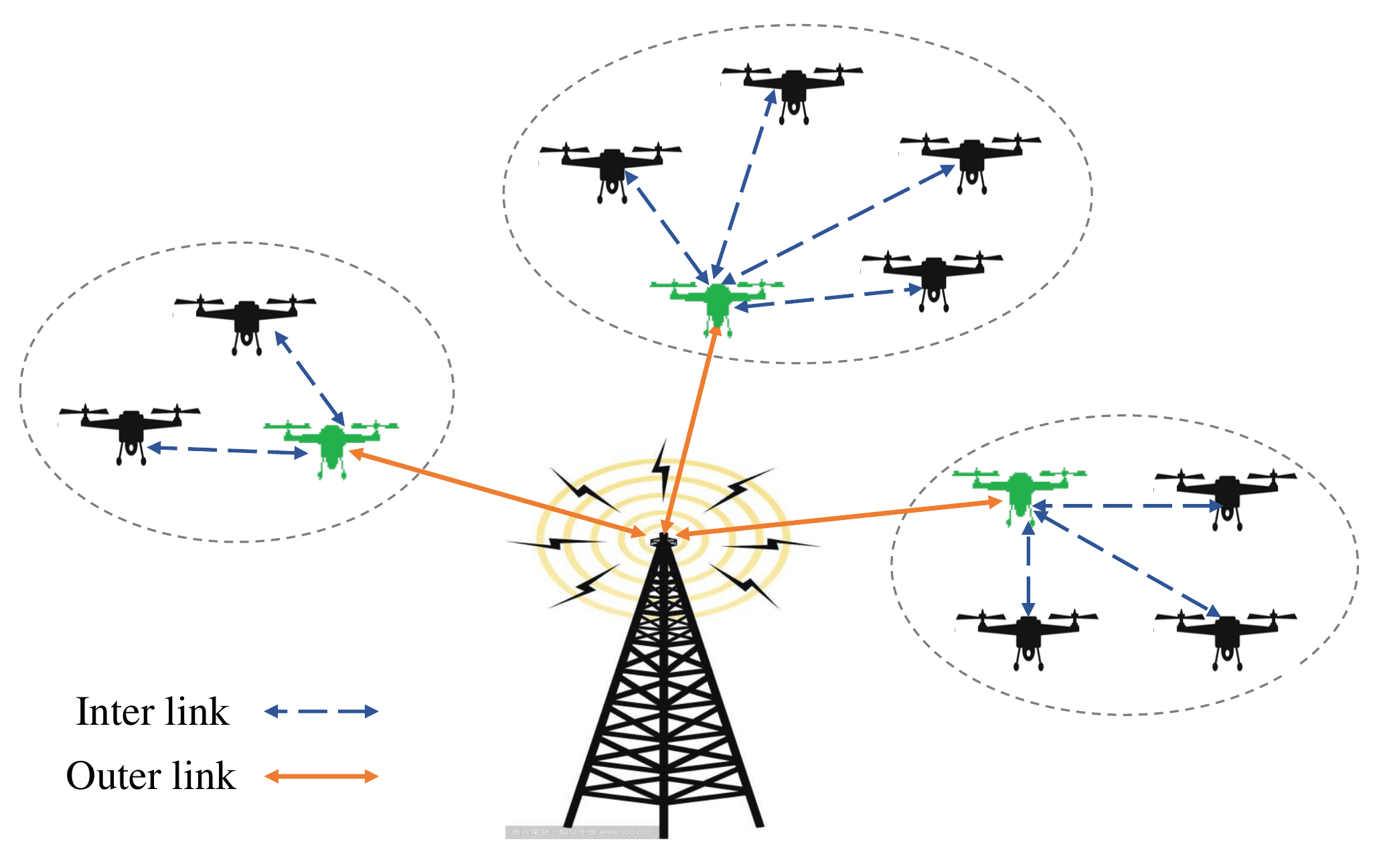}
 \captionsetup{margin=10pt,format=hang,justification=justified}      
 \caption{The system model of clustered mesh architecture UAVs network.}
\end{figure}
\subsection{Network Architecture}
We consider a mesh-structured UAVs network, where the UAV nodes act as clients (such as undertaking data-collection tasks) and attempt to convey information to others or to the GCS with limited number of channel resources. As shown in Fig. 1, the UAV nodes are randomly distributed in a 3-D space, and will form clusters according to their spatial positions \cite{breheny2003using}. A cluster unit consists of a cluster head (CH) and multiple cluster members (CMs). The CMs use inter-links to communicate with their corresponding CHs, while the CHs use outer-links to forward information to GCS. Compared with the outer-links, the length of the inter-links is much shorter, which results in reduced downlink bandwidth requirements and improved energy efficiency. Considering the energy limitation of the CHs as well as the QoS requirements of the CMs, the number of UAV nodes in a specific cluster is bounded by $C_{th}$.

Based on the above constructed UAV communication architecture, we denote the set of UAV nodes as $\mathcal{N}=\{1,2,...,N\}$, in which one UAV can be modeled by the Poisson Point Process (PPP). The full state of each UAV node is given by:
\begin{equation}
\Phi_n=\Big\{L_n,S_n,P_n,A_n,\Theta_n\Big\},n\in\mathcal{N}.
\end{equation}

Here $L_n=(x_n,y_n,z_n)$ is the spatial position of node $n$ \cite{li2015deep}. A Bernoulli random variable $S_n\in\{0,1\}$ characterizes the role mode of node $n$ in a cluster. Specifically, if node $n$ acts as a CH, $S_n=1$, otherwise, $S_n=0$. $P_n\in\{P_M,P_H\}$ is the transmission power of node $n$, which is determined by $S_n$, i.e.,
\begin{eqnarray}
P_n=\begin{cases}P_M,& S_n=0,\\P_H,& S_n=1.\end{cases}\nonumber
\end{eqnarray}

A Bernoulli random variable $A_n\in\{0,1\}$ represents the transmission state of node $n$, i.e., if node $n$ is active in the current time slot, $A_n=1$, else $A_n=0$. $\Theta_n\in\mathcal{H}$ denotes the corresponding CH of node $n$, and $\mathcal{H}=\{n\in\mathcal{N}|S_n=1\}$ is the CHs set. Based on the clustered structure of UAV networks, the set of UAV nodes $\mathcal{N}$ can be rewritten as $\mathcal{N}=\{\mathcal{C}_1,\mathcal{C}_2,...,\mathcal{C}_{|\mathcal{H}|}\}$, where $\mathcal{C}_i$ is the UAV sets of cluster $i$.

The POCs are assumed to support the communication of the clustered UAV networks. Denote the set of the available channels as $\mathcal{M}=\{1,...,m,...,M\}$, and the minimum channel separation for two channels to be regarded orthogonal as $\tau$. Therefore, the maximum number of the orthogonal channels can be expressed as $O_M=[\frac{M}{\tau}]$. For a specific channel $m$, the set of orthogonal channels is denoted as $\mathcal{M}^m_{OC}$, and the orthogonal channels set class of $\mathcal{M}$ is denoted as $\mathcal{M}_{OC}$, i.e. $\mathcal{M}_{OC}=\{\mathcal{M}^m_{OC}|m\in\mathcal{M}\}$. Specifications on the above parameters can be found in IEEE 802.11b/g standard. An illustrative partition of the POCs is shown in Fig. 2.
\begin{figure}[!t]           
 \centering
 \includegraphics[width=80mm]{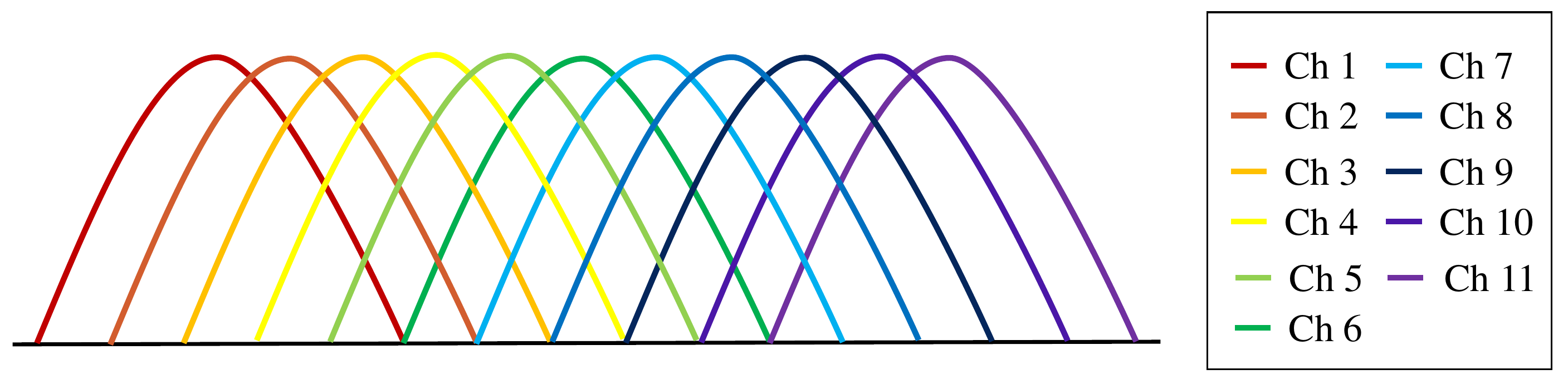}
 \captionsetup{margin=10pt,format=hang,justification=justified}      
 \caption{The partially overlapping channels(802.11b/g standard).}
\end{figure}

\subsection{Interference Model}
For a wireless network with MRMC, there are three different types of interferences which should be addressed to ensure the reliable transmission of network nodes: co-channel interference, adjacent channels interference, and self interference (e.g. when a single node utilizes two adjacent channels). To comprehensively describe the interference, a metric named \emph{Interference Factor} (IF) is recommended \cite{duarte2012partially}, which is defined as a ratio of spatial distance and \emph{Interference Range} (IR) between two nodes, and represents the effective channel separation level. To be specific, the IR refers to the minimum distance that two UAVs should obtain to avoid interference, and it is related with a channel separation factor $\delta=|i-j|$ (e.g. between channel $i$ and $j$). Based on the real measurements \cite{feng2008characterizing} and the scale-up degree \cite{duarte2012partially}, the relationships between IR and $\delta$ are given in Table I, in which the IR is measured in meters.
\renewcommand\arraystretch{1.5}
\begin{table}[h]
    \caption{Interference Range (IR)}
    \begin{center}
        \begin{tabular*}{8cm}{|p{0.7cm}<{\centering}|p{0.7cm}<{\centering}|p{0.7cm}<{\centering}|p{0.7cm}<{\centering}|p{0.7cm}<{\centering}|p{0.7cm}<{\centering}|p{0.7cm}<{\centering}|}
            \hline
            $\delta$ &0 &1 &2 &3 &4 &$\geqslant$5\\ \hline
            $IR(\delta)$ &132.6 &90.8 &75.9 &46.9 &32.1 &0\\
            \hline
        \end{tabular*}
    \end{center}
\end{table}

Denote the distance between two nodes $n_1$ and $n_2$ operating with channel $i$ and $j$ as $d_{n_1n_2}$, then the IF can be evaluated via the following three cases:
\begin{enumerate}[(1)]
 \item $IF(\delta,d_{n_1n_2})=0$, when $\delta\geqslant5$ or $d_{n_1n_2}>IR(\delta)$.

In this case, the nodes are assigned into orthogonal channels or have enough distance to avoid interference. Thus, there is no interference between them.
 \item $IF(\delta,d_{n_1n_2})=IR(\delta)/d_{n_1n_2}$, when $0\leqslant\delta<5$ and $0<d_{n_1n_2}\leqslant IR(\delta)$.

When two nodes occupy the overlapping channels, and meanwhile, the spatial distance between them is less than the IR, then the co-channel interference ($\delta=0$) or the adjacent channels interference ($0<\delta<5$) will arise.
 \item $IF(\delta,d_{n_1n_2})=\infty$, when $0\leqslant\delta<5$ and $d_{n_1n_2}=0$.

This situation corresponds to the self interference, which is excluded for its serious damages to the QoS. That is to say, two overlapping channels ($\delta<5$) will not be assigned to a single node.
\end{enumerate}

Premised on the above interference model to analyze POCs, we will design a robust channel allocation scheme to avoid co-channel and adjacent channels interference, thereby enhance the overall throughput of UAV communication networks in the dynamic and uncertain environments.

\subsection{Problem Formulation}
The high mobility, dynamic topology, intermittent links and varying link quality are the inherent properties of UAV communication networks, which need to be carefully addressed in practice.
To illustrate the dynamics of UAV nodes, we apply the Paparazzi mobility model (PMM), which is a stochastic mobility model to simulate UAVs behavior \cite{bouachir2014mobility}. The PMM consists of five possible movement types: \emph{Stay-At}, \emph{Way-point}, \emph{Eight}, \emph{Scan}, \emph{Oval}, which have different trajectories. Following a PMM, the spatial position $L_n$ of UAV node $n$ is changeable, making the link distance between the node $n$ and the destination node (the CHs or the GCS) time-variant. Thus, there would exist a discrepancy of link distance between the real value $D_n$ and the estimated value $\hat{D}_n$, i.e.,
\begin{equation}
|D_n-\hat{D}_n|=\Delta D_n,
\end{equation}
where $\Delta D_n$ is the uncertain boundary.

By introducing a scaling factor $\varepsilon_{n,m}$ \cite{feng2008how}, the channel gain of UAV node $n$, which is related with the link distance, is denoted as:
\begin{equation}
h_{n,m}=K\varepsilon_{n,m}(D_0/D_n)^{\varsigma_m}, \quad n\in\mathcal{N}, m\in\mathcal{M},
\end{equation}
where $K$ is a constant to reflect the influence of antenna gain and the average channel attenuation, $D_0$ is a reference distance which is fixed to be 1$\thicksim$10m indoors and 10$\thicksim$100m outdoors, and $\varsigma_m$ is the path loss exponent.

Due to the dynamic characteristics of UAV networks, the CSI tends to be uncertain and can be hardly estimated. Denote the imperfect estimation of channel gain as $\hat{h}_{n,m}$ \cite{li2016deep}, then the uncertain channel gain can be expressed as:
\begin{equation}
h_{n,m}=\hat{h}_{n,m}+\Delta h_{n,m},
\end{equation}
where $\Delta h_{n,m}$ accounts for a bounded error \cite{hasan2015distributed}.

Moving on, the signal to interference and noise ratio (SINR) $\gamma_{n,m}$ of node $n$, when accessing channel $m$, can be described as:
\begin{equation}
\gamma_{n,m}=\frac{P_nh_{n,m}}{I_n+\sigma^2_m},\quad n\in\mathcal{N}, m\in\mathcal{M},
\end{equation}
where $\sigma^2_m$ is the variance of additive white Gaussian noise (AWGN) on the channel $m$, and the $I_n$ represents the interference, which is given by:
\begin{eqnarray}
I_n=\sum\limits_{i\in\mathcal{N}\setminus n|d_{n,i}<IR(\delta)}P_nh_{n,m}.
\end{eqnarray}

(1)\textbf{ Achievable Rate} Let $\bm{s}_n=\{s^1_n,s^2_n,...,s^{\alpha_n}_n\}\in\mathcal{S}_n$ denote the allocated channel set of node $n$, where $\alpha_n$ is the number of the allocated channel, and $\mathcal{S}_n$ is the feasible channels set of node $n$. According to Shannon's capacity formula, the achievable data rate of node $n$ is:
\begin{equation}
R_n(\bm{s}_n)=A_n\sum^{\alpha_n}_{i=1}r_n(s^i_n)=A_n\sum^{\alpha_n}_{i=1}B\log_2(1+\gamma_{n,s^i_n}),
\label{Rn}
\end{equation}
where $B$ is the bandwidth of each channel, and $r_n(s^i_n)$ is the data rate in channel $s^i_n$.

(2)\textbf{ Generalized Throughput} Considering the mesh architecture of UAV networks, the performances of node $n$ should not be assessed only by its data rate, but also the topology structure. Instead of the achievable data rate $R_n$, we introduce another more comprehensive indicator $T_n$ \cite{duarte2012partially} to reformat the throughput of node $n$, whereby the previously defined IF as well as network connectivity are also taken into consideration, i.e.
\begin{equation}
T_n(\bm{s}_n)=\beta_n\frac{\sum^{\alpha_n}_{i=1}\frac{r_n(s^i_n)}{IF_{s^i_n}+1}}{\kappa_n}, s^i_n\in\bm{s}_n, n\in\mathcal{N},
\label{Tn}
\end{equation}
where $\beta_n\in\{0,1\}$ denotes a connectivity factor, $IF_{s^i_n}$ is the IF when node $n$ occupying channel $s^i_n$, and $\kappa_n$ is the hop count form node $n$ to the destination receiver.

Denote the channel allocation pattern of all UAV nodes as $\bm{s}=\{\bm{s}_1,\bm{s}_2,...,\bm{s}_N\}$, with the above two performance metrics, i.e. the achievable rate $R_n$ and the generalized throughput $T_n$, then the global utility $U(\bm{s})$ of all UAVs can be given by:
\begin{equation}
U(\bm{s})=\left\{
\begin{aligned}
&\sum_{n\in\mathcal{N}}R_n(\bm{s}_n)=\sum_{n\in\mathcal{N}}A_n\sum^{\alpha_n}_{i=1}r_n(s^i_n),\\
&\sum_{n\in\mathcal{N}}T_n(\bm{s}_n)=\sum_{n\in\mathcal{N}}\frac{\beta_n}{\kappa_n}\sum^{\alpha_n}_{i=1}\frac{r_n(s^i_n)}{IF_{s^i_n}+1}
\end{aligned}
\right.
\end{equation}

The purpose of channel resources management is thereby to optimize the network utility $U(\bm{s})$ (According to the practical application to determine which one is preferred, the rate $R_n$ or the throughput $T_n$), by carefully allocating the POCs with the following constraints.
\begin{enumerate}[(1)]
 \item \textbf{Total power constraint}:
 \begin{equation}
\alpha_nP_n\leqslant P^n_{max},
 \end{equation}
 where $P^n_{max}$ denotes the maximum transmission power of node $n$.
 \item \textbf{QoS constraint}:
 \begin{equation}\label{Rth}
R_n(\bm{s}_n)\geqslant R^n_{th},
 \end{equation}
 where $R^n_{th}$ is the minimum transmit data rate for node $n$ to maintain QoS requirement of diverse applications.
 \item \textbf{Cluster size constraint}:
 \begin{equation}
|\mathcal{C}_i|\leqslant C_{th},
 \end{equation}
 where $C_{th}$ is the maximum number of UAV nodes that a cluster unit can accommodate.
 \item \textbf{Orthogonality Constraint}:
 \begin{equation}
\bm{s}_n\in\mathcal{M}_{OC}, 0\leqslant \alpha_n\leqslant O_M.
 \end{equation}
As we claimed before, the self interference would severely undermine the QoS. Therefore, multiple channels occupied by one UAV node should be orthogonal, obviously, whose number can't surpass $O_M$.
\end{enumerate}

Based on the above elaborations, the corresponding POCs allocation problem for the mesh-structured UAV networks can be mathematically formulated as:
\begin{align}
&\max U(\bm{s}),\\
\text{s.t.}\quad &C1:\alpha_nP_n\leqslant P^n_{max},\nonumber\\
&C2:R_n(\bm{s}_n)\geqslant R^n_{th},\nonumber\\
&C3:|\mathcal{C}_i|\leqslant C_{th},\nonumber\\
&C4:\bm{s}_n\in\mathcal{M}_{OC}, 0\leqslant \alpha_n\leqslant O_M.
\end{align}

Due to the NP-hard nature of the above problem, solving it in static networks with definite information is already challenging, not to mention taking the dynamic properties of UAV networks and additional complex constrains into considerations. To the best of our knowledge, a robust algorithm for POCs allocation, one that can efficient cope with the dynamics and uncertainties in the context of UAVs communication networks, has not been studied in the literature.

\section{Fuzzy Payoffs Game for POCs Allocation}
As shown by previous analysis and subsequent simulations, the intrinsic dynamics and uncertainties of the considered UAV networks would degenerate the performance of conventional crisp-game theoretical algorithms, for its sensitiveness to environmental variations. In this section, we will exploit fuzzy game theory \cite{zimmermann2011fuzzy} to reformulate the above optimization problem with uncertain information and multiple constraints. Specifically, we first summarize some basic definitions and notions of the conventional crisp-game theory and the fuzzy set theory. On this basis, we introduce the FPG concept to describe the POCs allocation problem in mesh UAV networks. Furthermore, the existence of the equilibrium solution for our established FPG is demonstrated.
\subsection{Game Theory}
\begin{myDef}[\textbf{Crisp Non-Cooperative Game}]
A crisp non-cooperative game is defined as $\mathcal{G}\triangleq\Big(\mathcal{N},\mathcal{S},\mathcal{F}(\bm{s})\Big)$, where:
\begin{itemize}
 \item $\mathcal{N}=\{1,2,...,N\}$ is the set of players (UAV nodes);
 \item $\mathcal{S}=\otimes\mathcal{S}_n$, $n\in\mathcal{N}$ is the set of strategy profiles of the game, where $\mathcal{S}_n$ is the set of strategies of the $n$th player;
 \item $\mathcal{F}=\{F_1,F_2,...,F_N\}$ is the set of payoff functions for the players.
\end{itemize}
\end{myDef}

Since the game $\mathcal{G}$ is non-cooperative, then only self-enforcing solutions can be reasonable and rational for it. The core concept of the non-cooperative game $\mathcal{G}$ is Nash equilibrium (NE), which is described as follows.
\begin{myDef}[\textbf{Nash Equilibrium}]
A strategy pattern $\bm{s}^*\in\mathcal{S}$ is called a NE of the game $\mathcal{G}$ if,
\begin{align}
F_n(\bm{s}^*_n,\bm{s}^*_{-n})\geqslant &F_n(\bm{s}_n,\bm{s}^*_{-n}),\nonumber\\
&\forall n\in\mathcal{N},~\forall \bm{s}_n,\bm{s}^*_n\in\mathcal{S}_n,~\bm{s}^*_{-n}\in\mathcal{S}_{-n},
\end{align}
where $\mathcal{S}_{-n}$ is the strategies sets of all players, except the $n$th player, and $\bm{s}^*_{-n}$ is an element of $\mathcal{S}_{-n}$.
\end{myDef}

\subsection{Fuzzy Set Theory}
\begin{myDef}[\textbf{Fuzzy Number}]
A real fuzzy number $\tilde{a}$ is precisely described as any fuzzy subset on the space of real numbers $\mathscr{R}$, whose membership function $\mu_{\tilde{a}}(x)$ satisfies the following conditions:
\begin{itemize}
 \item $\mu_{\tilde{a}}(x)$ is a continuous mapping from $\mathscr{R}$ to the closed interval $[0,1]$.
 \item $\mu_{\tilde{a}}(x)$ is constant on $[-\infty,a_1]\cup[a_4,+\infty]$ and $[a_2,a_3]$. Specifically, $\mu_{\tilde{a}}(x)=0$, $\forall x\in[-\infty,a_1]\cup[a_4,+\infty]$ and $\mu_{\tilde{a}}(x)=1$, $\forall x\in[a_2,a_3]$.
 \item $\mu_{\tilde{a}}(x)$ is strictly increasing and continuous over $[a_1,a_2]$, and strictly decreasing and continuous over $[a_3,a_4]$.
\end{itemize}

Here $a_1$, $a_2$, $a_3$ and $a_4$ are real numbers satisfying $a_1<a_2\leqslant a_3<a_4$.
\end{myDef}

The membership function $\mu_{\tilde{a}}(x)$ gives a quantitative description of the fuzzy number $\tilde{a}$, which is the basic concept of fuzzy mathematics. Here, we take the triangular fuzzy number (TFN) $\tilde{a}=(a,l,r)$ for example, whose membership functions $\mu_{\tilde{a}}(x)$ is given by:
\begin{eqnarray}
\mu_{\tilde{a}}(x)=\begin{cases}\displaystyle\frac{x-a+l}{l},& \qquad x\in[a-l,a],\\
                                \displaystyle\frac{a+r-x}{r},& \qquad x\in[a,a+r],\\
                                0,& \qquad \text{else},
                   \end{cases}
                   \label{TFN}
\end{eqnarray}
where $a$, $l$ and $r$ are all real numbers. An illustration of the membership function of TFN is shown in Fig. 3.
\begin{figure}[!t]           
 \centering
 \includegraphics[width=60mm]{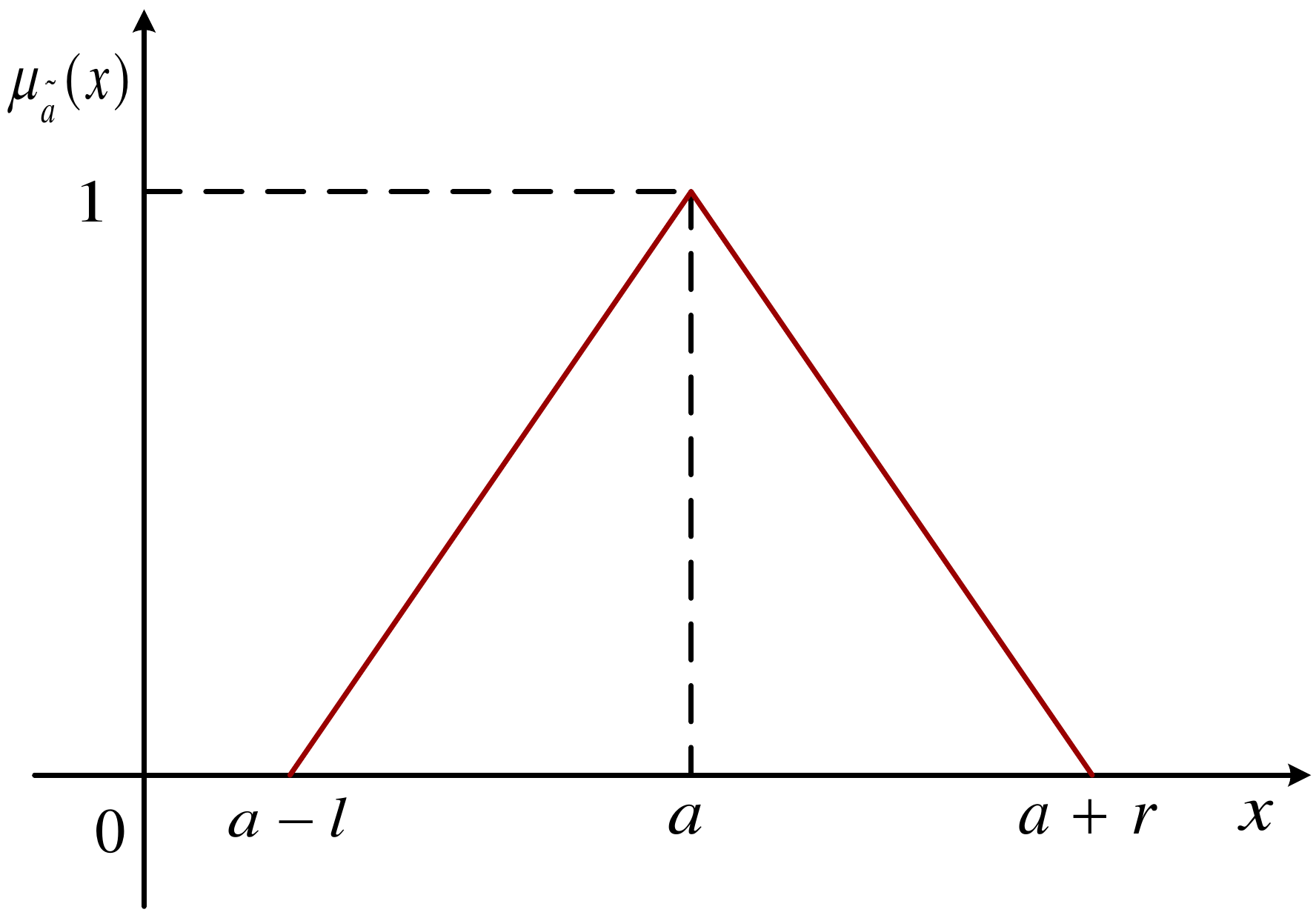}
 \caption{Membership function of $\tilde{a}$.}
\end{figure}

The operations of the fuzzy number obey the following lemma.
\begin{myLem}
{Let $\tilde{a}_1=(a_1,l_{a_1},r_{a_1})$, $\tilde{a}_2=(a_2,l_{a_2},r_{a_2})$ represent TFNs, $\nu$ is a real number. It holds that:
\begin{itemize}
 \item $\tilde{a}_1+\tilde{a}_2=(a_1+a_2,l_{a_1}+l_{a_2},r_{a_1}+r_{a_2})$;
 \item $\nu\tilde{a}_1=(\nu a_1,\nu l_{a_1},\nu r_{a_1})$;
 \item $\tilde{a}_2$ dominates $\tilde{a}_1$ (denoted by $\tilde{a}_2 \gtrapprox \tilde{a}_1$) if and only if $\max\{l_{a_2}-l_{a_1},0\}\leqslant a_2-a_1$ and $\max\{r_{a_1}-r_{a_2},0\}\leqslant a_2-a_1$.
\end{itemize}}
\end{myLem}

Since fuzzy numbers represent ambiguous numeric values, it is difficult to rank them according to their magnitude. Various methods of fuzzy numbers ranking have been developed. In \cite{lee1999method}, the method of evaluating fuzzy numbers with the satisfaction function (SF) and the viewpoint, and then ranking the numbers on the basis of their relative indexes of the evaluation values is introduced. The definitions of the SF, the viewpoint, the evaluation value, and the relative index are presented as follows.
\begin{myDef}[\textbf{Satisfaction Function}]
The SF between two fuzzy number $\tilde{a}$ and $\tilde{b}$ is defined as:
\begin{equation}
SF(\tilde{a}<\tilde{b})\triangleq\frac{\displaystyle\int^{+\infty}_{-\infty}\displaystyle\int^{y}_{-\infty}\mu_{\tilde{a}}(x)\odot\mu_{\tilde{b}}(y)dxdy}{\displaystyle\int^{+\infty}_{-\infty}\displaystyle\int^{+\infty}_{-\infty}\mu_{\tilde{a}}(x)\odot\mu_{\tilde{b}}(y)dxdy},\tag{18a}
\end{equation}
\begin{equation}
SF(\tilde{a}>\tilde{b})\triangleq\frac{\displaystyle\int^{+\infty}_{-\infty}\displaystyle\int^{+\infty}_{y}\mu_{\tilde{a}}(x)\odot\mu_{\tilde{b}}(y)dxdy}{\displaystyle\int^{+\infty}_{-\infty}\displaystyle\int^{+\infty}_{-\infty}\mu_{\tilde{a}}(x)\odot\mu_{\tilde{b}}(y)dxdy},\tag{18b}
\end{equation}
where the operator $\odot$ is a T-norm, without losing of generality, here we employ the commonly used multiplication operator. $SF(\tilde{a}<\tilde{b})$ represents the possibility that fuzzy number $\tilde{a}$ is smaller than $\tilde{b}$. Similarly, $SF(\tilde{a}>\tilde{b})$ represents the possibility that $\tilde{a}$ is larger than $\tilde{b}$.
\end{myDef}

\begin{myDef}[\textbf{Viewpoint}]
For a fuzzy numbers $\tilde{a}$, a fuzzy set $\tilde{b}$ which satisfies the following conditions is a viewpoint:
\begin{itemize}
 \item $\sup(\tilde{a})\subseteq\sup(\tilde{b})$, where $\sup(\tilde{a})=\{x|\mu_{\tilde{a}}(x)\neq0\}$;
 \item $\displaystyle\int^{+\infty}_{-\infty}\mu_{\tilde{b}}(x)\mathrm{d}x$ exists and it is not zero.
\end{itemize}

The fuzzy set $\tilde{b}$ is a viewpoint, which is used for evaluating the fuzzy numbers and can be broadly divided into three categories: \emph{optimistic} \emph{neutral} and \emph{pessimistic}. The second condition is added so that a viewpoint can be applicable to the SF.
\end{myDef}
\begin{myDef}[\textbf{Evaluation Value}]
On the basis of the interpretation of the SF, the evaluation value of fuzzy number $\tilde{a}$ in a viewpoint $\tilde{b}$, $E_{\tilde{b}}(\tilde{a})$ is given by:
\setcounter{equation}{18}
\begin{equation}
E_{\tilde{b}}(\tilde{a})=SF(\tilde{a}>\tilde{b})
\label{E}.
\end{equation}
\end{myDef}
\begin{myDef}[\textbf{Relative Index}]
The relative index of the fuzzy number $\tilde{a}$ in the viewpoint $\tilde{b}$, $V_{\tilde{b}}(\tilde{a})$, which shows how close $\tilde{a}$ is to the one having the best evaluation in viewpoint $\tilde{b}$, is defined as:
\begin{equation}
V_{\tilde{b}}(\tilde{a})=\frac{E_{\tilde{b}}(\tilde{a})}{\max_{\tilde{a}\in\tilde{\mathcal{A}}}E_{\tilde{b}}(\tilde{a})},
\label{I}
\end{equation}
where $\tilde{\mathcal{A}}$ is the set of fuzzy numbers.
\end{myDef}

\subsection{Fuzzy Payoffs Game}
In order to deal with the POCs allocation problem with uncertain information in dynamic UAV communication networks, we map the channel assignment problem into a fuzzy-logic space rather than an observational space. Then, we employ the TFN $\tilde{H}_{n,m}=(\hat{h}_{n,m},\Delta h^l_{n,m},\Delta h^r_{n,m})$ to describe the uncertain channel gains. Without losing of generality, we presume the left deviation $\Delta h^l_{n,m}$ and the right deviation $\Delta h^r_{n,m}$ of the TFN $\tilde{H}_{n,m}$ are equal, i.e. $\Delta h^l_{n,m}=\Delta h^r_{n,m}=\Delta h_{n,m}$. Recall that $\Delta h_{n,m}$ gives the bounded estimation error of channel gains. For clarity, we denote
\begin{equation}
\tilde{H}_{n,m}=(\hat{h}_{n,m},\Delta h_{n,m}).
\end{equation}

Expanding each component of the crisp game $\mathcal{G}$ to a fuzzy set would lead to a fuzzy game. In this paper, we assume the players set $\mathcal{N}$ and the strategy profiles $\mathcal{S}$ are definite, whereas the payoff of each player which influenced by the TFN $\tilde{H}_{n,m}$ is a fuzzy number.
\begin{myDef}[\textbf{Fuzzy Payoff Function}]
The FPF of each UAV node is defined as the uncertain achievable rate $R_n(\bm s_n, \tilde{H}_{n,m})$ or generalized throughput $T_n(\bm s_n, \tilde{H}_{n,m})$ of node $n$, which would become fuzzy numbers due to the fuzzy space projecting, i.e.,
\begin{equation}
\tilde{F}_n(\bm s_n,\bm s_{-n},\tilde{H}_{n,m})\triangleq\left\{
\begin{aligned}
&R_n(\bm s_n, \tilde{H}_{n,m}),\\
&T_n(\bm s_n, \tilde{H}_{n,m}).
\end{aligned}
\right.
\label{F}
\end{equation}

In this regards, the utility of each node is not only affected by the action taken of all nodes, but also by the dynamic and uncertain environments.
\end{myDef}

With the formulated FPF, we now present a FPG to characterize the problem of POCs allocation in UAV communication networks with indefinite CSI.
\begin{myDef}[\textbf{Fuzzy Payoffs Game}]
The FPG is defined as:
\begin{equation}
\label{G}
\tilde{\mathcal{G}}\triangleq\Big(\mathcal{N},\mathcal{S},\tilde{\mathcal{F}}(\bm{s},\bm{\tilde{H}})\Big),
\end{equation}
where $\mathcal{N}$ and $\mathcal{S}$ are identical with that in the crisp game $\mathcal{G}$. $\tilde{\mathcal{F}}(\bm{s},\bm{\tilde{H}})=\{\tilde{F}_n(\bm{s}, \tilde{H}_{n,m})|n\in\mathcal{N},m\in\mathcal{M}\}$ is the FPF set, and $\bm{\tilde{H}}=(\tilde{H}_{n,m}|n\in\mathcal{N},m\in\mathcal{M})$ is the vector of uncertain channel gains modeled by fuzzy number.
\end{myDef}

Based on the above analysis, the problem on eq. (14) constrained by eq. (15) can be reformulated as a FPG, in which the players attempt to find an appropriate channel selection pattern to maximize their fuzzy payoffs, i.e.,
\begin{align}
\bm s^*=\arg&\max_{\bm s_n\in\mathcal{S}_n}\tilde{F}_n(\bm s_n,\bm s_{-n},\tilde{H}_{n,m}),~\forall n\in\mathcal{N},\\
\text{s.t.}\quad &C1:\alpha_nP_n\leqslant P^n_{max},\nonumber\\
&C2:R_n(\bm{s}_n)\geqslant R^n_{th},\nonumber\\
&C3:|\mathcal{C}_i|\leqslant C_{th},\nonumber\\
&C4:\bm{s}_n\in\mathcal{M}_{OC}, 0\leqslant \alpha_n\leqslant O_M.
\end{align}

The property of the above designed FPG is investigated in the following subsection.

\subsection{Analysis of Fuzzy Nash Equilibrium}
As with the crisp game, the fuzzy game has also a NE concept, which is referred to as FNE \cite{chakeri2013fuzzy}. The definition of FNE is presented as follows.

\begin{myDef}[\textbf{Fuzzy Nash Equilibrium}]
A strategy pattern $\bm s^*\in\mathcal{S}$ is called a FNE of the fuzzy game $\tilde{\mathcal{G}}$ if,
\begin{align}
&\tilde{F}_n(\bm s^*_n,\bm s^*_{-n},\tilde{H}_{n,m})\gtrapprox\tilde{F}_n(\bm s_n,\bm s^*_{-n},\tilde{H}_{n,m}),\nonumber\\
&\forall n\in\mathcal{N},~\forall \bm s_n,\bm s^*_n\in\mathcal{S}_n,~\bm s^*_{-n}\in\mathcal{S}_{-n}.
\end{align}
\end{myDef}

\begin{myThe}
{\it There exists a FNE solution for the formulated FPG in eq. (\ref{G}).}
\end{myThe}
\begin{proof}
In order to demonstrate the existence of FNE, we first present the definition and the property of fuzzy bi-matrix games.

\begin{myDef}[\textbf{Fuzzy Bi-matrix Game}]
A fuzzy bi-matrix game $\tilde{\mathcal{G}}_B$ is defined as a bi-matrix game, which involves two players with fuzzy payoffs \cite{cunlin2011nash}, i.e.,
\begin{equation}
\tilde{\mathcal{G}}_B=\Big(\text{I},\text{II},\mathcal{S}_{\text{I}},\mathcal{S}_{\text{II}},\tilde{\textbf{F}}_{\text{I}},\tilde{\textbf{F}}_{\text{II}}\Big),
\end{equation}
where $\mathcal{S}_{\text{I}}$ and $\mathcal{S}_{\text{II}}$ are the sets of the strategies of Player \text{I} and Player \text{II}, respectively.
\begin{equation}       
\setlength{\arraycolsep}{1.5pt}
\tilde{\textbf{F}}_{n}=
\left[                 
  \begin{array}{cccc}   
    \tilde{F}_{11} & \tilde{F}_{12} & \cdots & \tilde{F}_{1M}\\  
    \tilde{F}_{21} & \tilde{F}_{22} & \cdots & \tilde{F}_{2M}\\  
    \vdots & \vdots & \vdots & \vdots\\
    \tilde{F}_{M1} & \tilde{F}_{M2} & \cdots & \tilde{F}_{MM}\\  
  \end{array}
\right]_{M\times M} n=\text{I}, \text{II},               
\end{equation}
is the payoffs matrix of the players. Each element $\tilde{F}_{m,m'}$ specifies the attained fuzzy payoffs, when Player \text{I} adopts the strategy $m$ while Player \text{II} adopts the strategy $m'$.
\end{myDef}

One key property of fuzzy bi-matrix games is characterized by the following lemma.
\begin{myLem}
{A fuzzy bi-matrix game has at least one FNE solution, if there exists a subset $\mathcal{N}_0 \subset\mathcal{N}$ such that the function $\sum_{n\in\mathcal{N}_0}\tilde{F}_n(\bm s_n,\bm s_{-n},\tilde{H}_{n,m})$ is convex on $\mu_{\tilde{H}_{n,m}}(x)$ \cite{kacher2008existence}.}
\end{myLem}

Based on the favorable feature of fuzzy bi-matrix games, the mathematical induction (MI) is employed to analyze our formulated FPG, with which the existence of FNE can be guaranteed.

First, we consider the situation that the FPG consists of two players and present the following theorem.
\begin{myThe}
{\it There exists at least one FNE solution for the FPG $\tilde{\mathcal{G}}_2$ with two players.}
\end{myThe}
\begin{proof}
For the first condition in \textbf{Lemma 2}, intuitively,
\begin{equation}
\tilde{\mathcal{G}}_2=(1,2,\mathcal{S}_1,\mathcal{S}_2,\tilde{\textbf{F}}_1,\tilde{\textbf{F}}_2),
\end{equation}
is a fuzzy bi-matrix game.

For the second condition, here, we choose $\mathcal{N}_0=\{1\}$, then we have
\begin{equation}
\sum\nolimits_{n\in\mathcal{N}_0}\tilde{F}_n(\bm s_n,\bm s_{-n},\tilde{H}_{n,m})=\tilde{F}_1(\bm s_1,\bm s_2,\tilde{H}_{1,m}).
\label{sts}
\end{equation}
To distinguish the concave-convex quality of function $\tilde{F}_1(\bm s_1,\bm s_2,\tilde{H}_{1,m})$, we calculate its second derivative $\tilde{F}^{''}_1$, i.e.
\begin{equation}
\tilde{F}^{''}_1=\frac{\mathrm{d}^2\tilde{F}_1(\bm s_1,\bm s_2,\tilde{H}_{1,m})}{\mathrm{d}\tilde{H}^2_{1,m}}=-\chi\sum^{\alpha_1}_{i=1}\frac{\psi}{\tilde{H}_{1,m}},
\end{equation}
where $\chi$ and $\psi$ are constants. According to the definition of FPF in eq. (\ref{F}), we discuss the following two cases.
\begin{enumerate}[(1)]
 \item When $\tilde{F}_1(\bm s_1,\bm s_2,\tilde{H}_{1,m})=R_1(\bm s_1,\tilde{H}_{1,m})$, we have
 \begin{equation}
 \left\{
 \begin{aligned}
 &\chi=A_1\geqslant 0,\\
 &\psi=\frac{B}{ln2}>0,
 \end{aligned}
 \right.\Rightarrow\tilde{F}^{''}_1\leqslant 0.
 \end{equation}
 \item When $\tilde{F}_1(\bm s_1,\bm s_2,\tilde{H}_{1,m})=T_1(\bm s_1,\tilde{H}_{1,m})$, we have
 \begin{equation}
 \left\{
 \begin{aligned}
 &\chi=\frac{\beta_1}{\kappa_1}\geqslant 0,\\
 &\psi=\frac{B}{\text{ln}2 \times (IF_{s^i_1}+1)}>0,
 \end{aligned}
 \right.\Rightarrow\tilde{F}^{''}_1\leqslant 0.
 \label{ste}
 \end{equation}
\end{enumerate}
In the above two cases, the function $\sum\limits_{n\in\mathcal{N}_0}\tilde{F}_n(\bm s_n,\bm s_{-n},\tilde{H}_{n,m})$ is convex on $\mu_{\tilde{H}_{n,m}}(x)$.

Based on \textbf{Lemma 2} and the above elaborations, \textbf{Theorem 2} can be proved.
\end{proof}

Moving on, we execute the second step of the MI and provide the following theorem.
\begin{myThe}
{\it Assuming that the FPG $\tilde{\mathcal{G}}_N$ with $N$ players has FNE solutions, then there exists at least one FNE solution for the FPG $\tilde{\mathcal{G}}_N+1$ with $N+1$ players.}
\end{myThe}
\begin{proof}
Denote the $N+1$th player as $n_0$, the FPG $\tilde{\mathcal{G}}_N+1$ can be expressed as:
\begin{equation}
\tilde{\mathcal{G}}_N+1=\Big(n_0, \mathcal{N},\mathcal{S}_{n_0},\mathcal{S}_{\mathcal{N}},\tilde{\textbf{F}}_{n_0},\tilde{\textbf{F}}_{\mathcal{N}}\Big).
\end{equation}

Due to the existence of FNE of the FPG $\tilde{\mathcal{G}}_N$, the players set $\mathcal{N}$ can be regarded as a whole unity. Therefore, the FPG $\tilde{\mathcal{G}}_N+1$ apparently turns into a fuzzy bi-matrix game, in which $n_0$ is Player I, and $\mathcal{N}$ is Player II.

After establishing the FPG $\tilde{\mathcal{G}}_N+1$ as a fuzzy bi-matrix game, we analyze the concave-convex property of the payoff functions $\tilde{F}_{n_0}$ and $\tilde{F}_{\mathcal{N}}$.

For the payoff function $\tilde{F}_{n_0}(\bm s_{n_0},\bm s_{\mathcal{N}},\tilde{H}_{n_0,m})$ of Player I, let $\mathcal{N}_0=\{n_0\}$, by performing the same steps (from eq. (\ref{sts}) to eq. (\ref{ste})) in \textbf{Theorem 2}, it can be proved convex on $\mu_{\tilde{H}_{n_0,m}}(x)$.

The payoff function $\tilde{F}_{\mathcal{N}}$ of Player II is the sum payoffs of all player $n$, $n\in\mathcal{N}$, i.e.
\begin{equation}
\tilde{F}_{\mathcal{N}}=\sum\nolimits_{n\in\mathcal{N}}\tilde{F}_n(\bm s_n,\bm s_{-n},\tilde{H}_{n,m}).
\end{equation}
Obviously, $\tilde{F}_{\mathcal{N}}$ is a convex function on $\mu_{\bm{\tilde{H}}}(x)$, since $\tilde{F}_n(\bm s_n,\bm s_{-n},\tilde{H}_{n,m})$ is convex on $\mu_{\tilde{H}_{n,m}}(x)$, $n\in\mathcal{N}$. Recall that $\bm{\tilde{H}}$ denotes the vector of uncertain channel gains.

On the basis of \textbf{Lemma 2} and the above analysis, \textbf{Theorem 3} can be proved.
\end{proof}
Combining \textbf{Theorem 2} and \textbf{Theorem 3}, \textbf{Theorem 1} can be proved.
\end{proof}

After we have demonstrated the existence of FNE for our formulated FPG, what remains to solve is how to achieve the equilibrium solution. It should be pointed out that the procedure of identifying FNE of fuzzy games is far more complex than finding NE in crisp games, which, for example, involves the fuzzy number ranking, and would be influenced by the membership function of fuzzy number as well as the viewpoint of players, rendering most existing learning methods invalid. Thus, we need to design a new learning algorithm to cope with the fuzzy parameters, with which the robust accessing and optimal allocation can be implemented in dynamic UAVs environments.

\section{Global Optimization in Dynamic UAV Communication Networks}
In order to solve the concerned FPG, in this section we will introduce a robust fuzzy-learning algorithm for distributed POCs allocation, and then demonstrate its convergence property.

\begin{figure}[!t]           
 \centering
 \includegraphics[width=88mm]{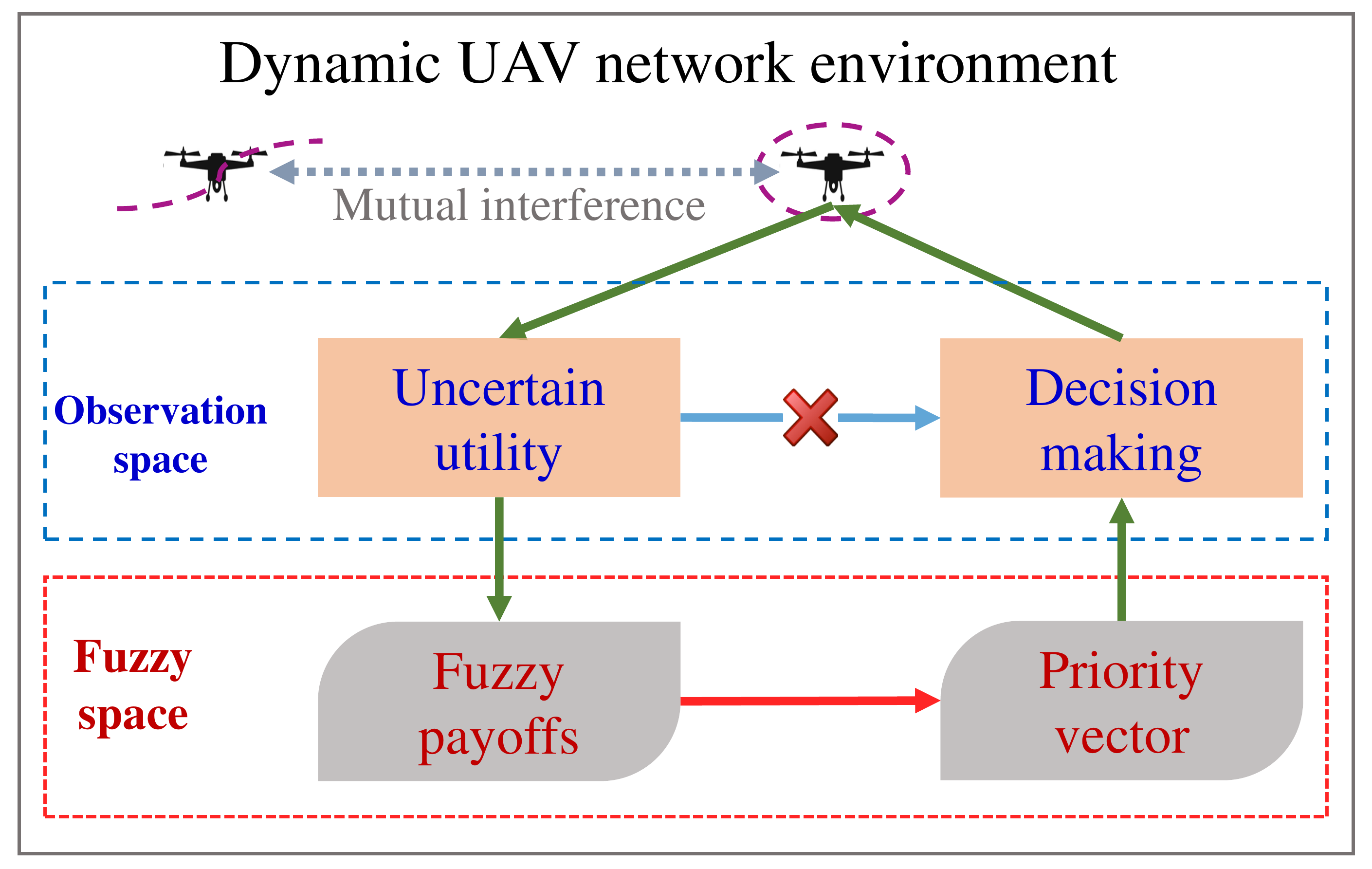}
 \caption{The framework of the fuzzy-learning algorithm}
\end{figure}
\subsection{Fuzzy-Learning Algorithm}
In UAV communication networks, UAV nodes would experience dynamic environments and suffer from varying coupling interference in most cases, i.e. their utilities become fluctuated and uncertain. Existing crisp-game theoretical learning approaches, depending on the accurate rewards (e.g. the channel capacity or SINR) to make decisions and update strategies in an observational space, are vulnerable for the encountered dynamics, hence the convergence can be hardly guaranteed.

Instead of an observational space, we implement the learning and updating in a mapped fuzzy space, whereby the priority vector of actions can be acquired by resorting to the fuzzy payoffs rather than the real-value crisp payoffs.
A remarkable advantage of our developed learning algorithm is that, by introducing a mapped fuzzy-space and the fuzzy-domain learning, it allows for a desensitization of fluctuated utility and thereby is capable of combating environmental changes and ensuring robust access.
A conceptional algorithm flow is shown in Fig. 4.

Based on the above elaborations, our proposed robust fuzzy-learning algorithm for POCs assignment in dynamic UAV networks is summarized in \textbf{Algorithm 1}. With the fuzzy-space interpolation and processing, a general action-taken strategy is adopted as follows:
\begin{equation} \label{sn}
\left\{
\begin{aligned}
&s^1_n(k+1)=\arg\max\bm{w}_n,\\
&s^i_n(k+1)\in\mathcal{M}_{OC}\Big(s^1_n(k+1)\Big),i=2,...,\alpha_n,\\
&\bm s_n(k+1)=\Big\{s^1_n(k+1),...,s^{\alpha_n}_n(k+1)\Big\}.
\end{aligned}
\right.
\end{equation}

It is worth to highlight that the updating rules contain two aspects, of which the first part corresponds to the best response with a fuzzy logic, while the second part is added to fulfill the QoS requirement (e.g. the orthogonality constraints).

As mentioned, we assume the channel gain $h_{n,m}(k)$ varies fast between two adjacent slots, and the strategy updating at $k$ slot is based on the current uncertain utility. The termination criterion of \textbf{Algorithm 1} is that, the difference of UAV node $n$'s utility between two adjacent iteration slots is less than a predefined threshold $\Delta R$ (data rate criterion) or $\Delta T$ (throughput criterion), i.e.,
\begin{equation}
\label{cri}
\left\{
\begin{aligned}
&R_n\Big(\bm s_n(k+1)\Big)-R_n\Big(\bm s_n(k)\Big)<\Delta R, \quad n\in\mathcal{N},\\
&T_n\Big(\bm s_n(k+1)\Big)-T_n\Big(\bm s_n(k)\Big)<\Delta T, \quad n\in\mathcal{N}.
\end{aligned}
\right.
\end{equation}

\begin{algorithm}[!t]
\caption{The Fuzzy-Learning Algorithm}
\label{The FGDA Algorithm}
\begin{algorithmic}
\STATE \textbf{Input}: Each UAV node $n$, $n\in\mathcal{N}$.
\STATE \textbf{Initialization}: The random channel selection pattern $\bm s_n$ of UAV node $n$.
\STATE \emph{\textbf{Step 1}}: At iteration slot $k$, each node $n$ obtains the fluctuated achievable rate (the generalized throughput) of different channel according to eq. (\ref{Rn}) (eq. (\ref{Tn})).
\STATE \emph{\textbf{Step 2}}: The node $n$ employs fuzzy space mapping to combat the dynamics of UAV networks and obtains the priority vector $\bm{w}_n$ of all channels.
\STATE \emph{\textbf{Step 3}}: Evaluate the actions according to $\bm{w}_n$, and update the selection strategy $\bm s_n$ in line with eq. (\ref{sn}).
\STATE \emph{\textbf{Step 4}}: Calculate the data rate of node $n$ according to eq. (\ref{Rn}), and compare it with the QoS constraint $R^n_{th}$. If the condition, i.e. eq. (\ref{Rth}) is satisfied, output the current selection strategy; otherwise, continue \emph{\textbf{Step 5}}.
\STATE \emph{\textbf{Step 5}}: Choose another channel among the set $\mathcal{M}_{OC}\Big(s^1_n(k+1)\Big)$. Repeat the \emph{\textbf{Step 4}}, until the QoS is implemented or the available channel is empty.
\STATE \emph{\textbf{Step 6}}: If the eq. (\ref{cri}) holds, stop; otherwise, go to \emph{\textbf{Step 1}}.
\STATE \textbf{Output}: The optimal channel selection pattern $\bm s^*_n$ of UAV node $n$.
\end{algorithmic}
\end{algorithm}

\subsection{Priority Vector}
From eq. (34) and the proposed \textbf{Algorithm 1}, it is noted that the priority vector $\bm{w}_n=\{w^m_n|n\in\mathcal{N}, m\in\mathcal{M}\}$ is the cornerstone of fuzzy-space learning for UAV node $n$, as far as the main purpose of POC channel assignment is concerned, which accounts for the important degree of channels with a viewpoint $\tilde{v}_n$ of node $n$ and also satisfies the normalizing condition $\sum_{m\in\mathcal{M}}w^m_n=1$, $w^m_n\geqslant0$. In the following, we will present a least deviation algorithm \cite{xu2005least}, in order to quantify the fuzzy payoffs $\tilde{F}_n(\bm s_n,\bm s_{-n},\tilde{H}_{n,m})$ and finally calculate the priority vector $\bm{w}_n$ via fuzzy logic.
\begin{algorithm}[!t]
\caption{The Least Deviation Algorithm}
\label{The Least Deviation Algorithm}
\begin{algorithmic}
 \STATE \emph{\textbf{Step 1}}: Rank the fuzzy number payoffs.
 \begin{itemize}
 \item Define a fuzzy number $\tilde{v}_n$ as the viewpoint of node $n$;
 \item Evaluate $E_{\tilde{v}_n}\Big[\tilde{F_n}(\bm s_n,\bm s_{-n},\tilde{H}_{n,m})\Big]$ for all fuzzy number payoffs using eq. (\ref{E});
 \item Calculate the relative index $V_{\tilde{v}_n}\Big[\tilde{F_n}(\bm s_n,\bm s_{-n},\tilde{H}_{n,m})\Big]$ for all fuzzy number payoffs using eq. (\ref{I}).
 \end{itemize}
 \STATE \emph{\textbf{Step 2}}: Calculate the FPR matrix $\textbf{Q}_n=[q_{i,j}]$ as
                                \setcounter{equation}{38}
                                \begin{align}
                                 &q_{ij}= \nonumber \\&\begin{cases}\min\Big\{\zeta\Psi_{i,j}+(1-\zeta)\Lambda_{i,j}+0.5,1\Big\},&\Psi_{i,j}>0,\\
                                                     0.5,&\Psi_{i,j}=0, \\
                                                     1-\min\Big\{\zeta\Psi_{j,i}+(1-\zeta)\Lambda_{j,i}+0.5,1\Big\},&\Psi_{i,j}<0.
                                        \end{cases}
                                \end{align}
                                $\forall i,j\in\mathcal{M}$, where
                                \begin{eqnarray}
                                \Psi_{i,j}\triangleq V_{\tilde{v}_n}\left[\tilde{F_n}(\bm s,\tilde{H}_{n,i})\right]-V_{\tilde{v}_n}\left[\tilde{F_n}(\bm s,\tilde{H}_{n,j})\right]
                                \end{eqnarray}
                                is the absolute difference, and
                                \begin{align}
                                \Lambda_{i,j}&\triangleq\frac{V_{\tilde{v}_n}\left[\tilde{F_n}(\bm s,\tilde{H}_{n,i})\right]-V_{\tilde{v}_n}\left[\tilde{F_n}(\bm s,\tilde{H}_{n,j})\right]}
                                        {V_{\tilde{v}_n}\left[\tilde{F_n}(\bm s,\tilde{H}_{n,j})\right]}\nonumber\\
                                       &=\frac{\Psi_{i,j}}{V_{\tilde{v}_n}\left[\tilde{F_n}(\bm s,\tilde{H}_{n,j})\right]}
                                \end{align}
                                is the relative difference. The $\zeta$ is used to fluctuate the weight of the absolute difference $\Psi_{i,j}$ and the relative difference $\Lambda_{i,j}$.
 \STATE \emph{\textbf{Step 3}}: Set $k=0$, and initialize the priority vector and specify parameter $0<\eta\leqslant1$.
 \STATE \emph{\textbf{Step 4}}: Calculate the term below:
 \begin{equation}
 \varphi_i=\sum_{j\in\mathcal{M}}\left[g(q_{ij})\frac{w^j_n}{w^i_n}-g(q_{ji})\frac{w^i_n}{w^j_n}\right] ~\forall i\in\mathcal{M},
 \end{equation}
 where $g(q_{ij})=9^{2q_{ij}-1}$. If $|\varphi_m|\leqslant\eta$ for all $m\in\mathcal{M}$, stop; otherwise, continue to \emph{\textbf{Step 5}}.
 \STATE \emph{\textbf{Step 5}}: Find out the number $\lambda$ that maximizes $|\varphi_m|$, $m\in\mathcal{M}$, i.e., $\varphi_\lambda=\max_{m\in\mathcal{M}}\{|\varphi_m|\}$, and calculate:
 \begin{eqnarray}
 Y=\sqrt{\left.\left[\sum_{j\in\mathcal{M}\setminus\lambda}g(q_{\lambda,j})\frac{w^j_n}{w^\lambda_n}\right]\middle/\left[\sum_{j\in\mathcal{M}\setminus\lambda}g(q_{j,\lambda})\frac{w^\lambda_n}{w^j_n}\right]\right.},
 \end{eqnarray}
 and
 \begin{eqnarray}
 \phi_m=\begin{cases}Y\times w^m_n,& \qquad m=\lambda,\\
                     w^m_n,& \qquad m\neq\lambda.\\
        \end{cases}
 \end{eqnarray}
 \STATE \emph{\textbf{Step 6}}: Update the priority vector as follows:
 \begin{equation}
 w^m_n=\frac{\phi_m}{\sum_{i\in\mathcal{M}}\phi_i},\qquad\forall m\in\mathcal{M}, \forall n\in\mathcal{N}.
 \end{equation}
\end{algorithmic}
\end{algorithm}

We first rank the fuzzy number payoffs, and then introduce a fuzzy preference relation (FPR) to make a soft measurement of fuzzy numbers, with which the stable priority vector $\bm{w}_n$ of the actions can be derived. The FPR matrix of node $n$ is defined as $\textbf{Q}_n=[q_{i,j}]$ with complementary matrix properties:
\begin{equation}
\left\{
\begin{aligned}
&q_{ij}+q_{ji}=1,\\
&q_{ij}\geqslant0,     \qquad\qquad\forall i,j\in\mathcal{M},\\
&q_{ii}=0.5,
\end{aligned}
\right.
\tag{38}
\end{equation}
where $q_{ij}$ denotes the preference degree of the UAV node $n$ between channel $i$ and $j$.

On this basis, the priority vector can be determined by incorporating the viewpoint projection. Provided the used triangular fuzzy number, the schematic flow of a least deviation algorithm is then illustrated by \textbf{Algorithm 2}. By resorting to the fuzzy-logic to analyze mapped fuzzy payoffs, the fuzzy-space interpolation leads to a desensitization of fast changing environments and fluctuated utilities, hence, each player would learn smoothly dynamic environments and evolve steadily towards a satisfactory solution.

\subsection{Convergence of the Proposed Algorithm}
To demonstrate the convergence of our new fuzzy-learning algorithm for POCs allocations, we then present the following theorem.
\begin{myThe}
{\it The proposed \textbf{Algorithm 1} for the dynamic UAV networks is guaranteed to converge to a stable channel allocation profile, with which the maximal network throughput can be achieved.}
\end{myThe}
\begin{proof}
For our robust fuzzy-learning algorithm for distributed POCs allocations, in \emph{\textbf{Step 3}} of \textbf{Algorithm 1}, due to the best response of strategy updating with fuzzy logic, the reward of UAV node $n$ is always non-decreasing, and thus the network throughput will increase until the global stability is achieved. Owing to the limitation of channel resources and UAV nodes, on the other hand, the global QoS utility is up bounded. Thus, the proposed algorithm would converge finally to a stable channel allocation profile after finite iterations, i.e., no further throughput improvement can be made as the maximal network throughput would have been achieved. Based on the above statements, \textbf{Theorem 4} is proved.
\end{proof}

\section{Simulation Results}
In this section, numerical simulations are provided to demonstrate the performance of our robust fuzzy-learning algorithm in the context of self-adaption POCs allocation in mesh UAV communication networks.
In our following analysis, the size of 3-D space is configured to $200\times200\times200$ m$^3$. The maximum number of UAV nodes for a cluster is 6, i.e., $C_{th}=6$.
Transmission powers of CH and CM are $P_1=10$dBm and $P_0=-10$dBm, respectively. The available channels for UAV networks are specified by IEEE 802.11b/g standard, i.e. $|\mathcal{M}|=11$, $\tau=5$ and $O_M=3$.
In order to characterize different movement types, the normalized uncertain boundary of varying channel gains is $\Delta h_{n,m}/\hat{h}_{n,m}\in[10^{-3}, 1]$.
The AWGN variance and the pass-loss exponent of channels are set as $\sigma^2_m=-80$dBm and $\varsigma_m=2$, $m\in\mathcal{M}$, respectively.
Other constant parameters for implementing fuzzy-learning algorithm are set to $\zeta=0.5$ and $\eta=0.8$. The viewpoint of UAV node is assume to be \emph{neutral}.

In the following, we firstly explained the mesh-structured UAV network used in our simulations.
Then, the convergence performance of our proposed scheme is provided, and both achievable rate and generalized throughput of our fuzzy-learning algorithm are compared with that of its counterpart.
Finally, the system performance is demonstrated via two main metrics: the number of active links, and the network throughput with the assurance of QoS.
Note that, all numerical results are derived from 50 independently simulated UAV network topologies and 100 trials for each network topology.

\begin{figure}[!t]           
 \centering
 \includegraphics[width=75mm]{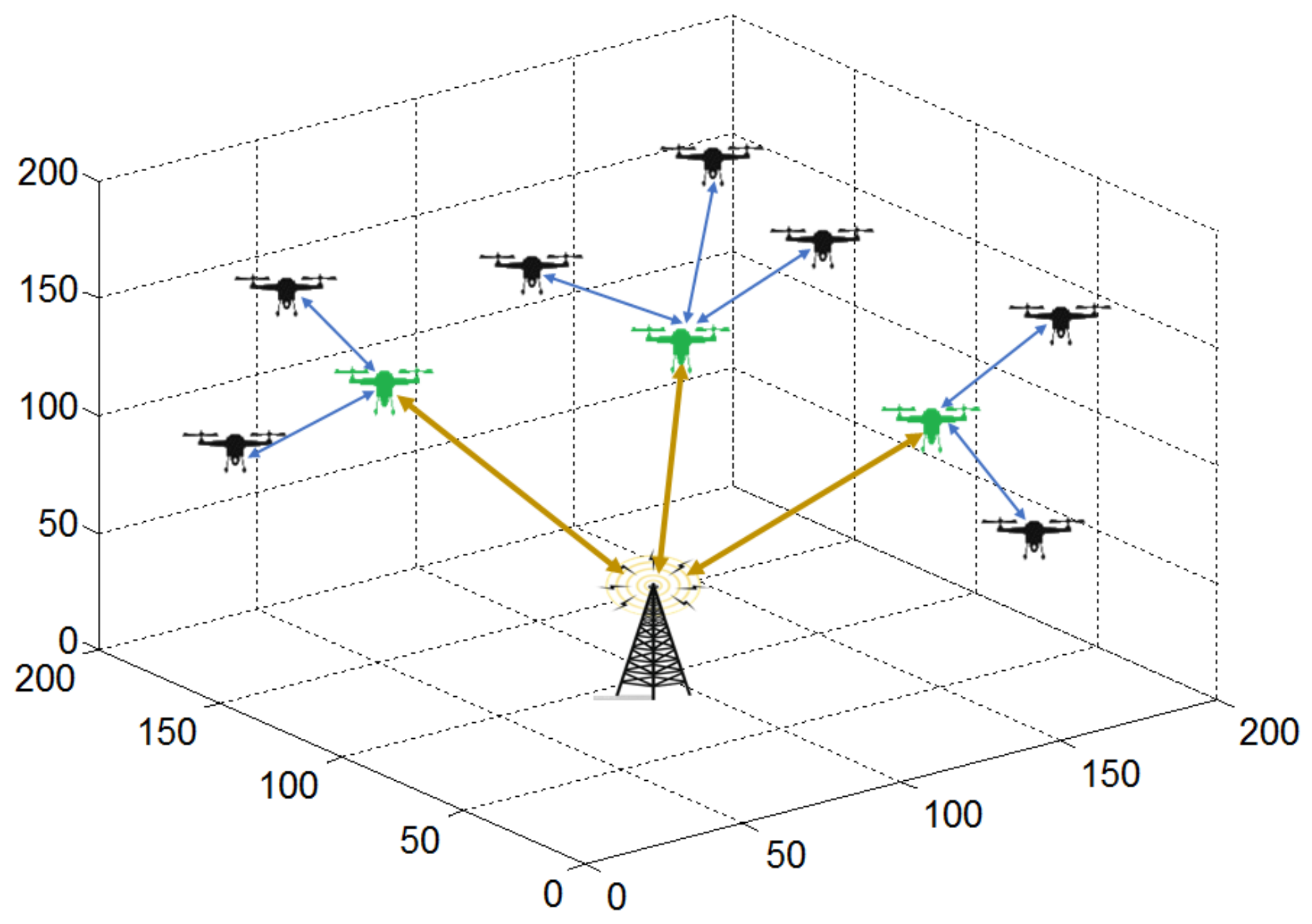}
 \caption{Illustrative diagrams of the mesh structured UAV network model}
\end{figure}
\subsection{Mesh UAV Network Model}
A diagram of simulated mesh UAV network is shown by Fig. 5, which involves a GCS and multiple UAV nodes. The network size is $N=10$, whereby total 10 UAV nodes formed 3 clusters (one may refer to some related works for cluster formulation algorithms, and here we just assume clustering is based on the spatial distances). So, the required number of outer-links is 3. For a star-structured UAV network, the number of required outer-links equals exactly to the number of UAV nodes. In comparison, the long-distance outer-links will be limited in a mesh UAV network, and more importantly, the QoS requirements can be fulfilled with a much lower transmission power (dominated by short-distance inter-links). Thus, in comparing with a star architecture, the mesh configuration will be more preferable for UAV communication networks.

\begin{figure}[!t]           
 \centering
 \includegraphics[width=75mm]{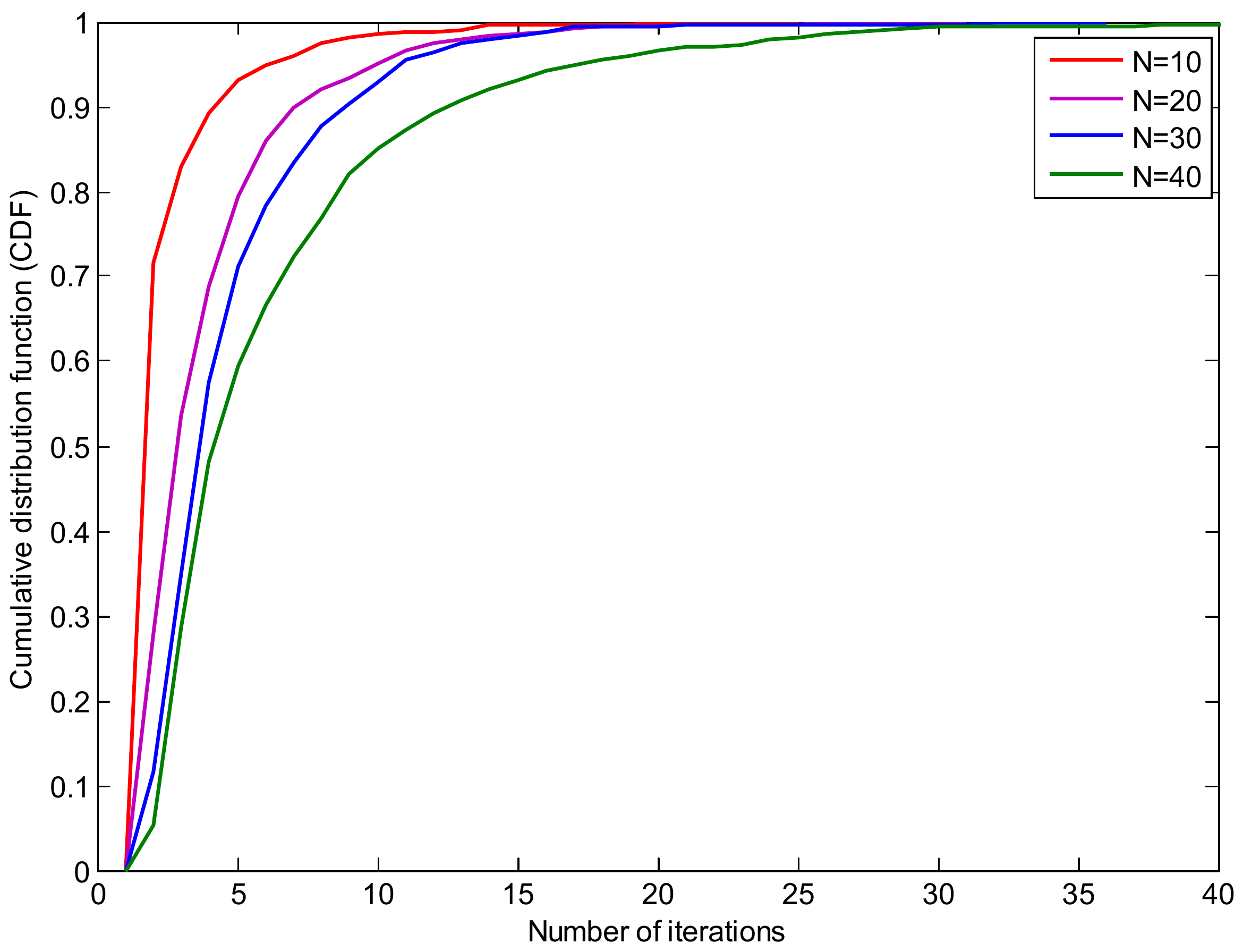}
 \caption{Convergence speed of our distributed POCs allocation algorithm with different UAV network size.}
\end{figure}
\subsection{Convergence Performance}
We then evaluate our proposed scheme in the context of dynamically uncertain UAV communication environments. First, the cumulative distribution function (CDF) of required iterations to achieve a satisfactory solution is presented in Fig. 6, which gives a indicator of convergence speed of our proposed fuzzy-learning scheme from a statistical perspective.
From numerical results, it seems that the iterations needed for convergence is positively related with the total number of UAV nodes. This is relatively easy to follow, i.e. the larger network size needs more iteration to achieve convergence. Besides, we note that the convergence of our proposed method is relatively rapid even for a larger network size, and the mean values of required iteration under different UAV network sizes ($N$=10, 20, 30, 40) are about 3, 4, 6 and 8, respectively. Such a rapid convergence property makes our scheme particularly attractive to the robust accessing in energy-constrained UAV communication networks.

\begin{figure}[!t]
\centering
\subfigure[] {\label{fig:a}
\centering
\includegraphics[width=75mm]{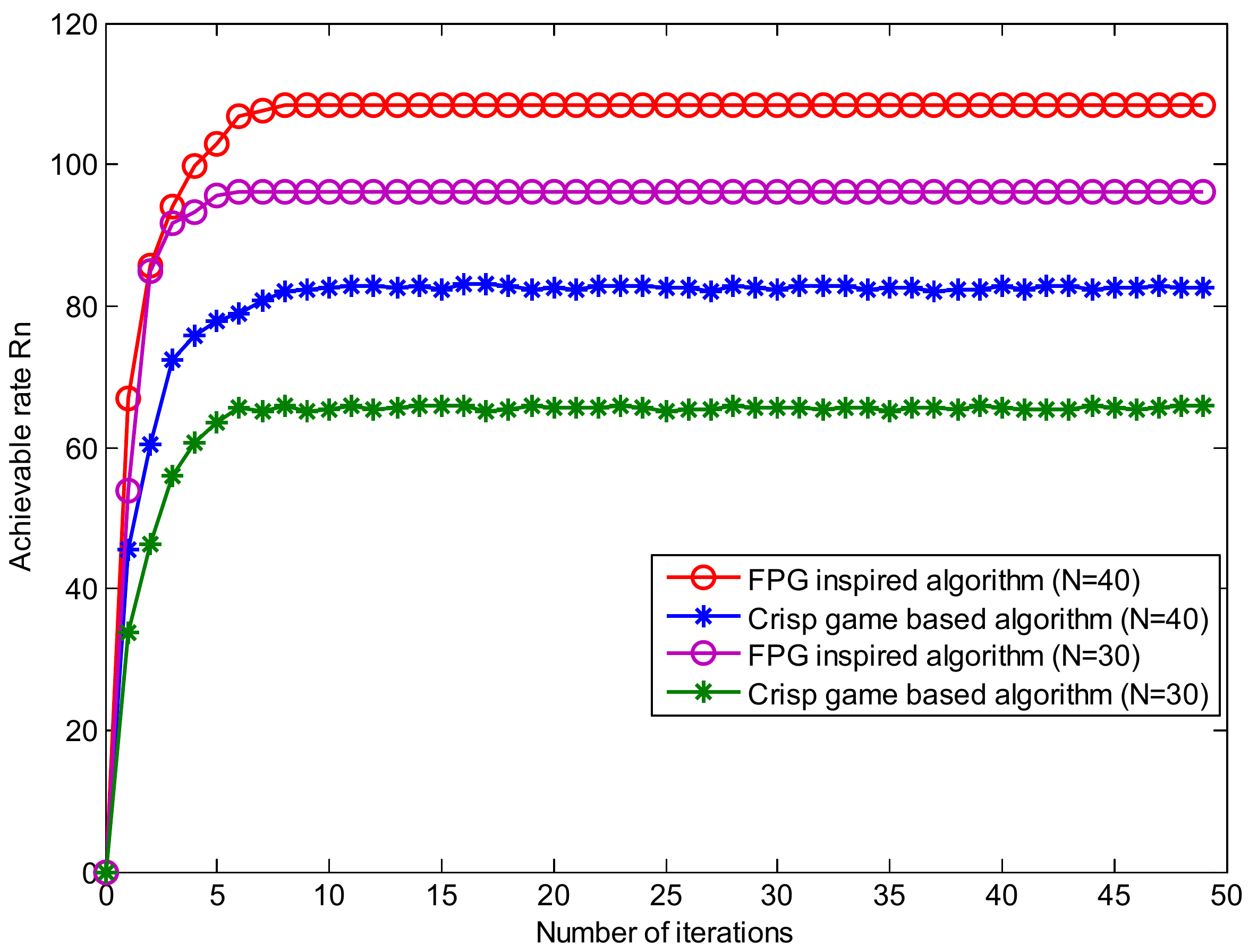}
}
\subfigure[] {\label{fig:b}
\centering
\includegraphics[width=75mm]{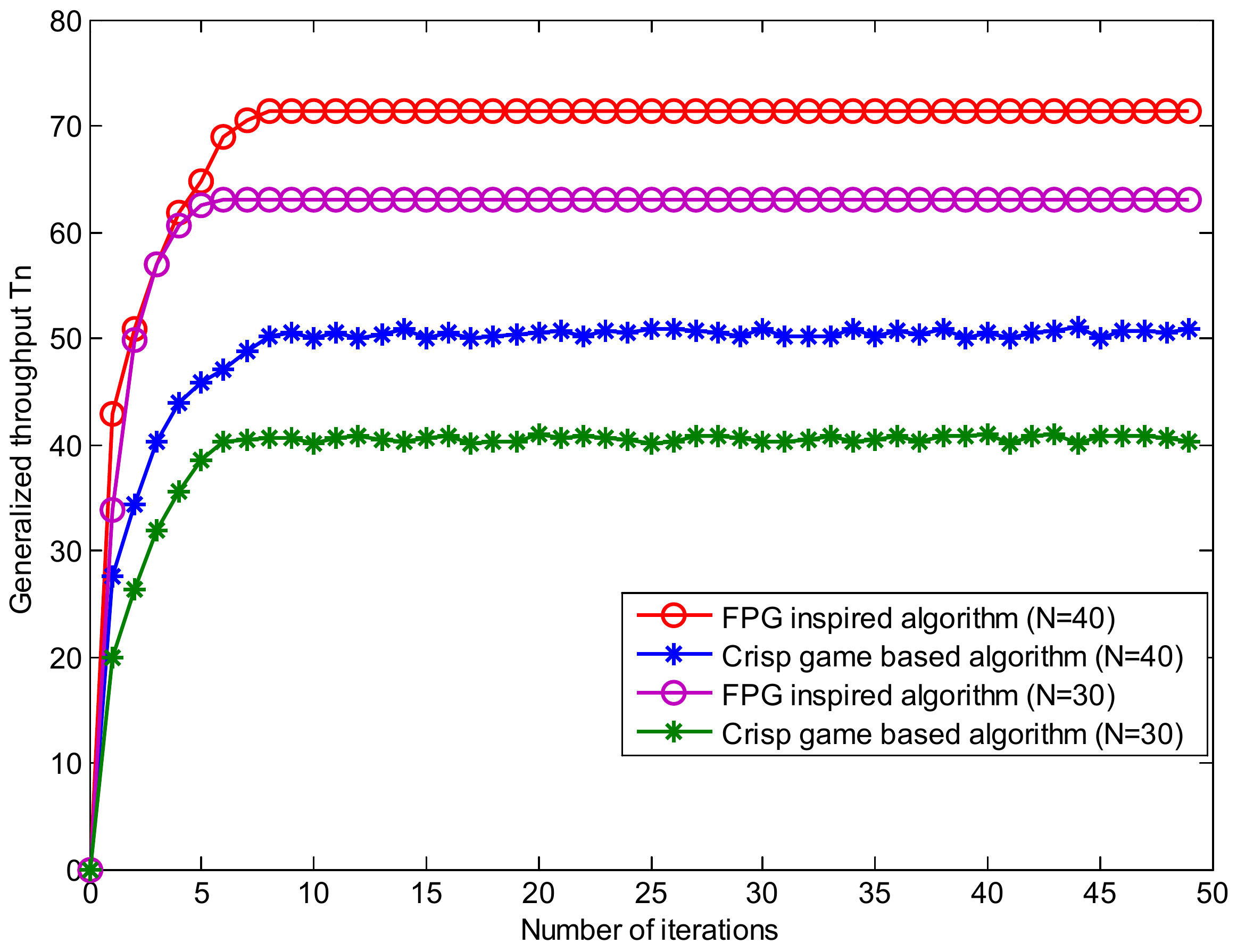}
}
\caption{The comparison of achievable rate and generalized throughput between fuzzy game based and crisp game based algorithms. (a) Achievable rate $R_n(\bm{s}_n)$; (b) Generalized throughput $T_n(\bm{s}_n)$.}
\end{figure}

After the convergence, comparative results for both achievable rate and generalized throughput between our proposed fuzzy-learning algorithm and its counterpart, i.e. the crisp-game theoretical algorithm, are plotted together in Fig. 7.
First, as far as these two different performance metrics are concerned, the generalized throughput would be less than the achievable rate under the same parameter configurations, which is consistent with our previous definitions (i.e., the generalized throughput is further weighted by the connectivity factor and interference factor). More importantly, with our new fuzzy-learning algorithm, the attained maximum utility of our proposed scheme will dramatically outperforms that of a conventional scheme, no matter what the performance metric is and how larger the UAV network size is. Taking the generalized throughput under $N=40$ for example, the total throughput of our new method converges to 71, whilst the crisp-game based algorithm can only approach 50. That is, the significant improvement, i.e. around $60\%$, is achieved by our proposed scheme.

The main reason is that, due to the lack of mechanisms to combat the randomly fluctuated utility, most conventional learning-based methods will be inevitably influenced by dynamic environments, and thereby fail to achieve the satisfactory solution. In contrast, by the implementing learning and updating in a mapped fuzzy space, our proposed scheme is basically immune to the involved dynamics and uncertainties, which is hence more competitive in identifying the optimal solution to POCs allocations and enables robust accessing even in dynamic UAV communication networks.

\subsection{System Performance}
We further study the performance of various multi-channel allocation algorithms under difference system configurations, i.e. the UAV network size $N$.
First, we are interested in the number of active links under different UAV network size.
Then, we present the comparative results, i.e. the achieved network throughput, of our proposed scheme and its counterpart methods, i.e. the crisp-game based algorithm and another random selection approach.

\subsubsection{Number of Active Links}
For our proposed fuzzy-learning scheme, the permitted number of parallel active links under different UAV network sizes is illustrated in Fig. 8.
It is found that, the numerical derived curve can be partitioned into an unsaturated regime and a saturated regime, with a regime bound of $N=70$.
In the left unsaturated regime, the number of active links would increases with total UAV nodes number $N$. In the right saturated regime, the number of active links would remain unchanged, which means the network capacity has an up-bound even considering the channel reuse, due to serious coupling interference.
For a specific parameter configuration, it is shown that the maximum number of active links is 35, and a maximal channel reuse ratio is around 3.

\subsubsection{Achieved Network Throughput}
Furthermore, we study the network throughput of our proposed scheme and the counterpart approaches under various UAV network size.
The expected network throughput is shown by Fig. 9, which simultaneously gives the mean value and the variance of network throughput with three allocation schemes.
It is observed that our proposed fuzzy-learning algorithm would significantly outperform the other two approaches no matter what the UAV nodes number is.
Our new method can attain the superior network performance (i.e. higher mean) with the more favorable stability (i.e. lower variance). In comparison, the other existing approaches may become less competitive, as far as robust accessing in dynamic environments is concerned, especially for a crisp-learning scheme with which the performance variance may even surpass its mean value (e.g. $N=10$ and $N=20$).

\begin{figure}[!t]           
 \centering
 \includegraphics[width=75mm]{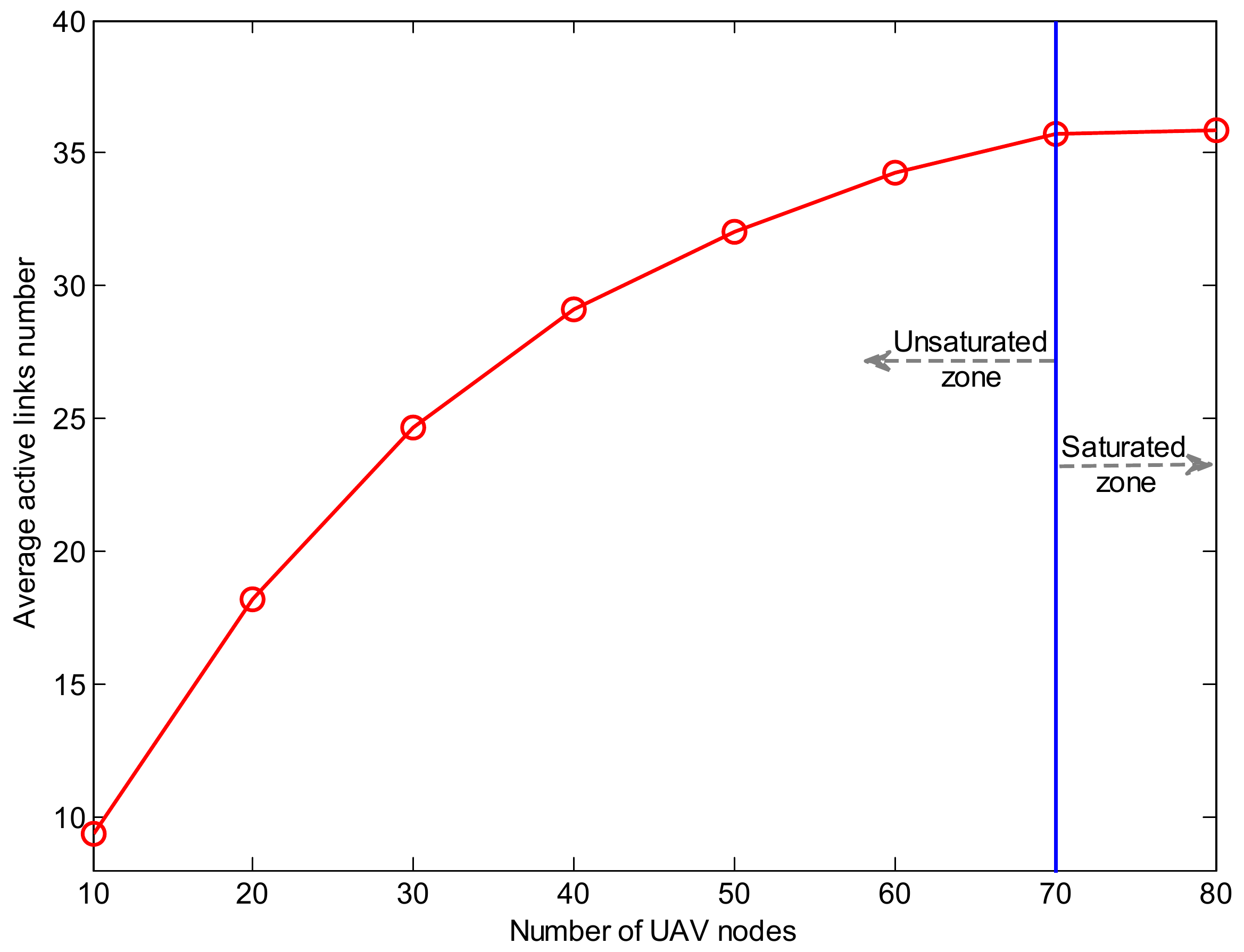}
 \caption{The number of active links}
\end{figure}

In particular, we noted from Fig. 9(a) that the network throughput achieved by a crisp-game algorithm is slightly greater than that of another random selection approach. In other words, a classical crisp-game algorithm, developed for most static environments (i.e. with time-independent channel gain and mutual interference), would become basically invalid, whose learning behavior was completely undermined by the constantly changing network topology and the time-varying channels or local utilities. In sharp contrast, despite the mobility of UAV nodes and the resulting dynamically uncertain CSI, with the fuzzy-space mapping, our proposed algorithm can still achieve the optimal POCs allocation by maximizing the network throughput.

In addition, from Fig. 9(b), we found that the performance variance of conventional crisp-learning algorithms is even inferior to that of a random selection approach, which further demonstrated the extreme vulnerability of a crisp-game method when handling the dynamically uncertain information in practical scenarios. And hence, it loses the effectiveness in channel allocations for dynamic UAV networks. This problem may hold for a large class of greedy-based learning schemes, whereby the temporal variations in local utility may trigger the impetuous response in updating strategies, causing sharp fluctuations in strategies and leading to random evolution behaviors.
By presenting the appealing fuzzy-space learning framework, our proposed scheme can cope with this challenging problem. As highlighted, it is capable of desensitizing the fluctuated utility, which thereby ensures robust accessing even in dynamic UAV networks, by producing a much lower variance in achieved performance.

\begin{figure}
\centering
\subfigure[] {\label{fig:a}
\centering
\includegraphics[width=75mm]{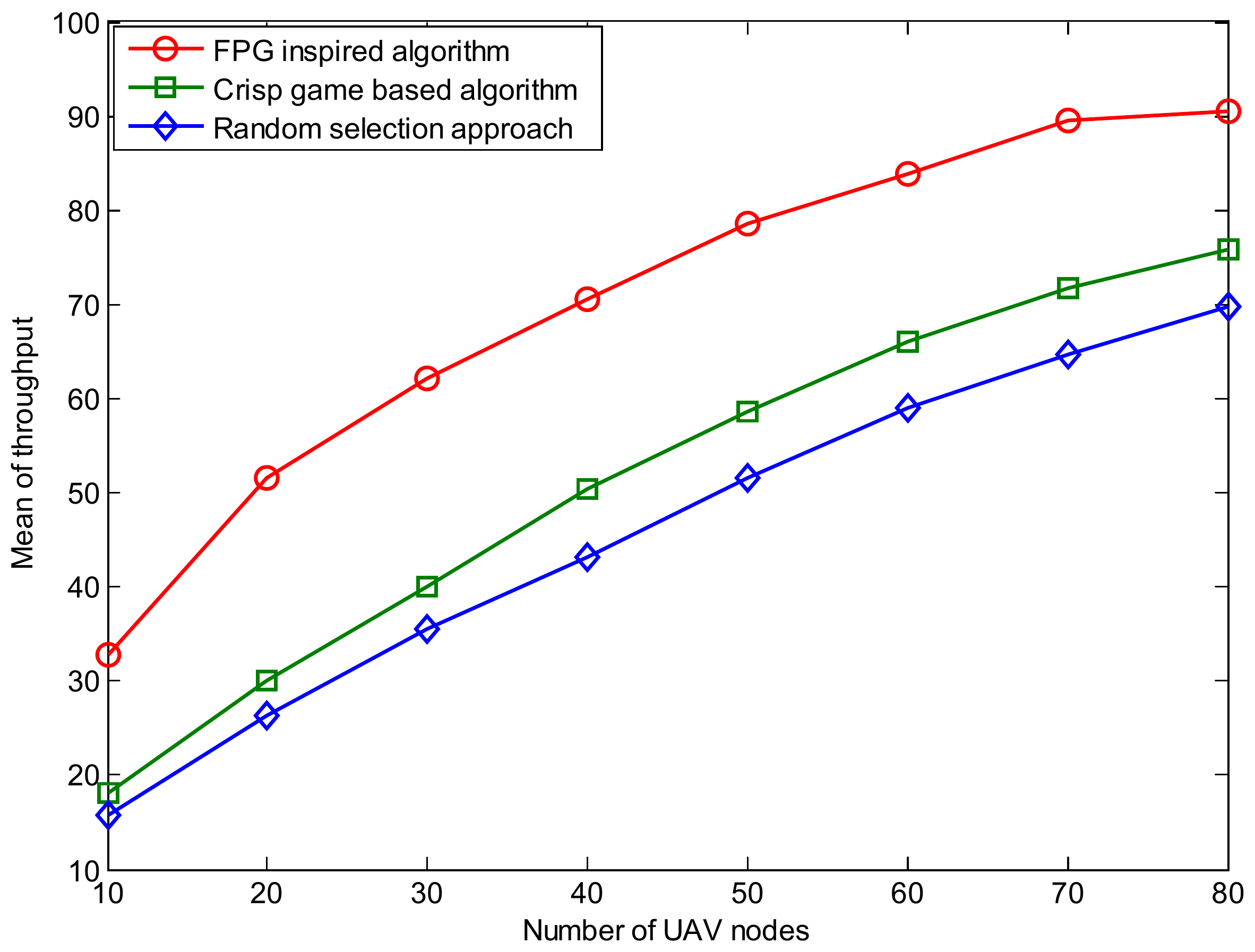}
}
\subfigure[] {\label{fig:b}
\centering
\includegraphics[width=75mm]{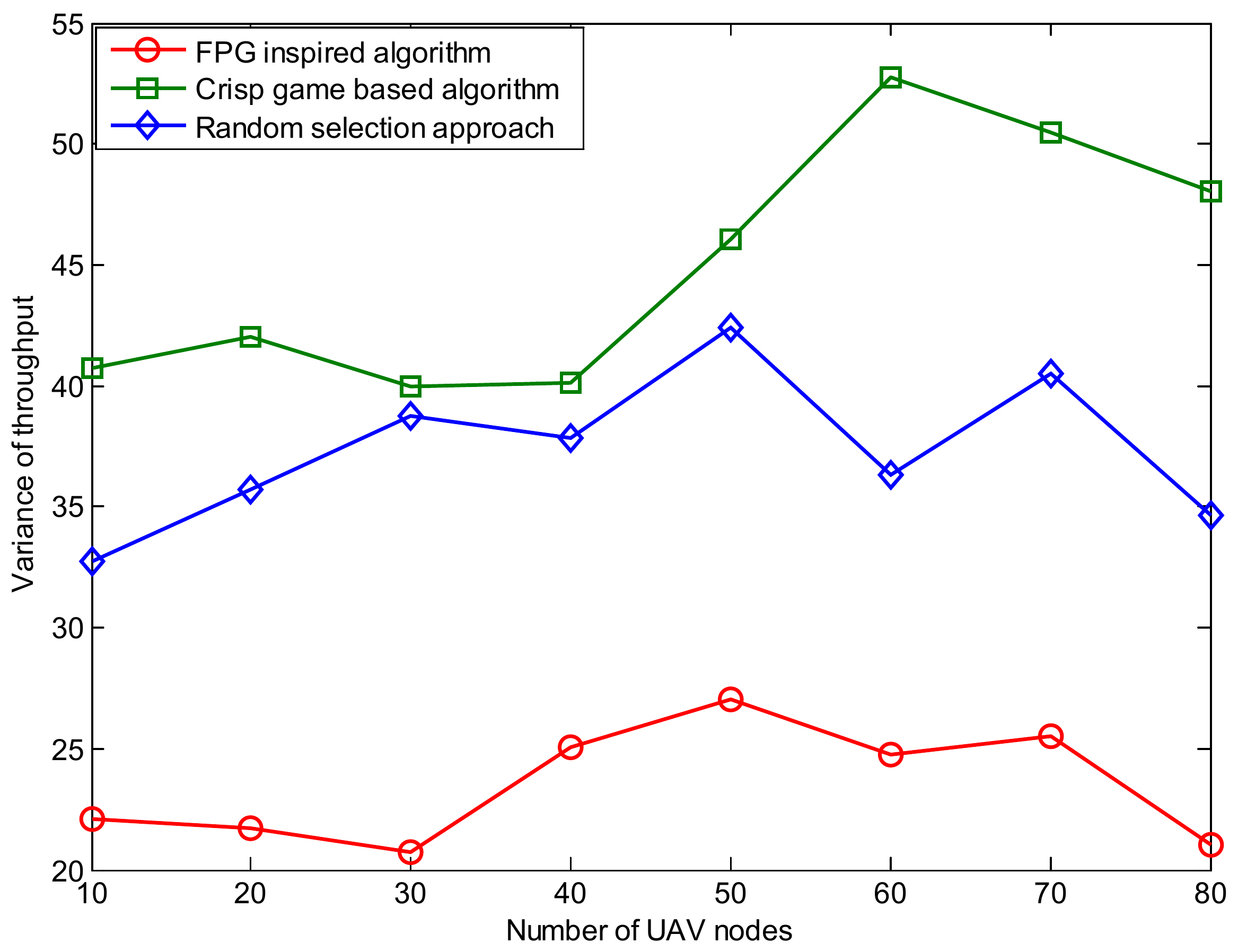}
}
\caption{The comparisons of the network throughput performances among the proposed scheme and the counterparts. (a) Throughput mean of different schemes; (b) Throughput variance of different schemes.}
\end{figure}
\section{Conclusion}
In this paper, the optimal POCs allocation with multiple constraints in dynamic mesh-structured UAV communication networks is studied.
By projecting the randomly fluctuated utility into another fuzzy-space, a robust fuzzy-learning paradigm for distributed POCs allocations is developed to cope with the major challenges of UAV communication networks, i.e. dynamically uncertain environment and energy-constrained deployment.
As opposite to the most existing learning schemes directly operated in the observational space, our proposed scheme implements the learning and updating in a mapped fuzzy space, whereby the fluctuated utilities are interpreted with fuzzy numbers and the decisions are made on the basis of the derived priority vectors.
Our new scheme is characterized by its appealing desensitization of dynamically uncertain CSI and fluctuated utilities, which would effectively combat the oscillation effects in decisions and thereby ensure the robust accessing even in dynamic UAV communication networks. Potential advantages of our algorithm are also demonstrated by numerical results.
Attributed to its notable stability and robustness, our proposed algorithm is capable of achieving the maximum network throughput and enhancing resource efficiency even in dynamically changing environments.
As a consequence, our proposed robust fuzzy-learning scheme will be of significant promise for the emerging UAV applications.

\section*{Acknowledgment}
This work was supported by BUPT Excellent Ph.D. Students Foundation under Grant No. CX2017209, and Natural Science Foundation of China (NSFC) under Grant No. 61471061.


%

\bibliography{myreference}
\bibliographystyle{IEEEtran}

\end{document}